\documentclass[aps, prx, superscriptaddress, reprint, floatfix, longbibliography]{revtex4-2}

\usepackage{graphicx}   
\usepackage[colorlinks = true, linkcolor=blue, citecolor=blue, urlcolor  = blue]{hyperref}
\usepackage{amsmath}    
\usepackage{amssymb}    
\usepackage{xcolor}
\usepackage{color}
\usepackage{physics}
\usepackage{tikz}
\usepackage{dsfont}
\usepackage{booktabs}
\usepackage{times}
\usepackage{comment}

\usepackage[capitalize]{cleveref}
\usepackage{quantikz}
\usepackage{svg}

\usetikzlibrary{shapes.geometric}
\usepackage{relsize}
\usepackage{tcolorbox}
\usepackage{enumerate}
\usepackage{amsthm}
\usepackage{orcidlink}

\crefname{equation}{Eq.}{equations}
\crefname{section}{Sec.}{sections}
\crefname{figure}{Fig.}{figures}
\crefname{appendix}{Appendix}{appendices}
\crefname{table}{Table}{tables}

\theoremstyle{plain}

\newtheorem{proposition}{Proposition}

\newtheorem{corollary}{Corollary}
\crefname{proposition}{Proposition}{Propositions}
\crefname{corollary}{Corollary}{Corollaries}
\AfterEndEnvironment{theorem}{\noindent\ignorespaces}
\AfterEndEnvironment{lemma}{\noindent\ignorespaces}
\AfterEndEnvironment{corollary}{\noindent\ignorespaces}

\AfterEndEnvironment{definition}{\noindent\ignorespaces}
\AfterEndEnvironment{proof}{\noindent\ignorespaces}

\newcommand{\UpsilonPT}{\Upsilon}

\begin{document}

\title{Quantum non-Markovian response spectra}

\author{Aaron Sander\,\orcidlink{0009-0007-9166-6113}}
\email{aaron.sander@tum.de}
\affiliation{Technical University of Munich, Munich, Germany}

\author{John F. Kam\,\orcidlink{0009-0008-7597-1094}}
\affiliation{Monash University, Melbourne, Australia}
\affiliation{Agency for Science, Technology and Research (A*STAR), Singapore}

\author{Gauthameshwar S.\,\orcidlink{0009-0005-9715-5758}}
\affiliation{Singapore University of Technology and Design, Singapore}
\affiliation{Centre for Quantum Technologies, Singapore}

\author{Robert Wille\,\orcidlink{0000-0002-4993-7860}}
\affiliation{Technical University of Munich, Munich, Germany}
\affiliation{MQSC GmbH, Munich, Germany}
\affiliation{Software Competence Center Hagenberg GmbH (SCCH), Hagenberg, Austria}

\author{Kavan Modi\,\orcidlink{0000-0002-2054-9901}}
\email{kavan\_modi@sutd.edu.sg}
\affiliation{Singapore University of Technology and Design, Singapore}
\affiliation{Centre for Quantum Technologies, Singapore}

\date{\today}

\begin{abstract}
Quantum systems are rarely isolated from their surroundings.
These hidden environments can retain information and return it later, producing correlated errors in quantum computers and complicating the interpretation of experiments.
On a quantum processor, for example, one job can therefore leave an environmental trace that affects the next, even after the qubits themselves have been reset.
Learning how this hidden influence persists and affects later behavior is therefore important, but the environment is often impossible to measure directly, and fully reconstructing its effects through tomography rapidly becomes impractical.
Here, we introduce a response matrix that uses controlled interventions and observations of the system alone to reveal the temporal structure of environmental memory.
The matrix records how future observations change across different controlled histories.
Its spectrum characterizes how memory is structured and how it persists under specified system resets, while suitable linear combinations of its entries can rule out specified classical models of that memory.
Because its history--future queries can be evaluated independently, numerical response matrices can be assembled in parallel.
We establish the framework theoretically, characterize response spectra in an interacting spin process, and certify quantum memory in a two-qubit benchmark.
This approach turns system observations into a probe of inaccessible environments and could support the diagnosis of correlated noise, tests of open system models, and memory-aware control design.
\end{abstract}

\maketitle

\section{Introduction}

Environmental memory makes a quantum system's future response depend on more than its present state.
It arises when degrees of freedom outside our direct access retain information about earlier interactions and return it later.
It can distort benchmarking signals \cite{white-nrp}, invalidate Markovian noise models, and frustrate control and error mitigation strategies that neglect correlations with the environment \cite{Harper2020,White2020}.
Structured vibrational and photonic environments can also sustain information backflow and long-lived coherence in molecular, biological, and solid-state systems \cite{Chin2013,GaugerPRL2011,Strathearn2018}.

Quantum processors bring this problem into sharp focus because they both suffer from environmental memory and provide a platform for studying it.
As hardware, they can retain unwanted information in spectators, leakage levels, material modes, or slowly evolving control parameters \cite{White2020,WhitePRXQ2022}, potentially compromising the performance of the device~\cite{Nielsen2021gatesettomography, FigueroaRomeroPRXQ2021}, and even compromising quantum error correction \cite{Kam2025, Kam2}.
As simulators, they can generate complex open dynamics in which a chosen subsystem interacts with a structured, evolving complement \cite{GarciaPerez2020,HanPRL2021,ErbanniPRA2023,Kim2023}.
In both roles, we usually access the process through controls and observations of the system rather than direct measurements of its environment.
One direct consequence is that the order of quantum jobs can matter.
A circuit can alter an inaccessible environment that later interacts with the next circuit, even after the accessible qubits have been reset \cite{White2020,PollockPRL2018}.

Additionally, memory is intrinsically a multi-time phenomenon.
The process tensor framework provides an operational language for describing it in open quantum dynamics and, more generally, quantum stochastic processes \cite{PollockPRA2018,PollockPRL2018,Milz2021}.
Building on quantum combs and operational causal modeling \cite{Chiribella2009,Costa2016}, a process tensor encodes how a process responds to arbitrary sequences of interventions.
It supports complementary questions about the strength of memory \cite{PollockPRA2018,Taranto2021}, its temporal extent or Markov order \cite{TarantoPRL2019,TarantoPRA2019,Taranto3}, its complexity \cite{PRXQuantum.4.020310,chang2026,wang2025}, and whether it is classical or quantum \cite{Giarmatzi2021,MilzSciPost2021,MilzPRX2020,Strasberg2019,Baecker2026}.
Machine learning and tensor network methods can learn both structured non-Markovian processes \cite{GuoPRA2020,Zhang2025,White2025}, and structured single-time quantum channels \cite{Torlai2023}.
As temporal representations, they also connect environmental memory to the classical cost of simulating open and many-body dynamics \cite{LerosePRB2021,FolignoPRX2023,Sonner2021,LinkPRL2024,Gribben2022,Cygorek2024ACE,KeelingPRX2026,Cygorek2022,Fux2023,Cygorek2025,Cygorek2024,Cygorek2024Sublinear,Ng2023,Park2024,Nguyen2024,Chen2024,Guo2024,Guo2024b,PRXQuantum.5.010314,Sun2025,WangPRL2024,fgc4-hgk1}. Complementary tensor network trajectory approaches also enable
scalable simulation of open quantum systems governed by
Markovian Lindblad dynamics \cite{Sander_TJM_2025}.

These advances leave a practical gap between detecting memory and reconstructing the entire process.
A causal break test can determine whether the future still depends on the history after the accessible system has been discarded and prepared anew \cite{PollockPRL2018}.
Detecting such a dependence shows that memory is present, but it does not determine how that dependence is organized, where it appears in time, or whether it survives later interventions.
At the other extreme, full process tensor tomography captures this information but generally requires resources that grow exponentially with the number of intervention times \cite{Milz2021}.
The unresolved question is whether selected intervention--response statistics can expose this structure directly.

Related predictive-state and realization approaches also organize dynamics through future responses.
Predictive-state representations encode a controlled stochastic system through predictions of future action--observation sequences \cite{Littman2001}, while quantum realization methods use Hankel constructions to infer the dimension of a hidden system from a local probe time series \cite{Sone2017} or classify finite-dimensional filters for continuously monitored bosonic systems \cite{Emerson2026}.
These methods seek a predictive state or dynamical realization, but share the broader goal of learning about inaccessible dynamics through observable responses.

In this work, we introduce a response matrix that extracts this missing structure directly from
selected causal break experiments. Rather than treating each history--future comparison as an
isolated memory test, the matrix organizes the dependencies measured across many histories and
future probes. Its singular spectrum separates the Markovian, memoryless baseline, in which the normalized future
response is independent of the tested history, from additional response patterns through which
memory remains visible. The normalized spectral weights then describe whether the measured response
is concentrated in one leading pattern or distributed across several.
Under the recovery conditions that we derive below, the response rank and spectrum match those of the
underlying process without requiring full reconstruction. 
Because the history--future queries are independent, the matrix can be assembled in parallel in simulation and constructed in the same way from experimental data or learned open system models.

The same response data can also certify quantum memory, provided that the selected probes support a valid witness with a negative expectation value.
Building on quantum-memory witnesses based on temporal Choi entanglement \cite{Giarmatzi2021,MilzSciPost2021}, we give a theoretical formulation of witness search as a semidefinite program within the span of the selected instrument probes, enforcing non-negativity on every process separable across the chosen history--future cut.
A negative witness value therefore excludes the adopted classical feed-forward model for the full process, without reconstructing its Choi operator.
The formulation applies to general finite-dimensional process tensors and specified probe families, such that whether it finds a certificate depends on the accessible space and the chosen positive map.
Failure to find a negative value remains inconclusive and does not establish that the memory is classical.
Our numerical demonstration evaluates an analytically supplied witness using the response data.

Crucially, this method does not require full process tomography or an informationally complete set of probes to obtain useful information about the non-Markovian spectra.
The chosen instruments define an accessible operator space in which the response matrix resolves memory, allowing the analysis to be tailored to available controls and measurements.
We establish conditions under which its rank and normalised singular spectrum recover those of the process tensor across the selected temporal cut, identifying when this operational description faithfully captures the underlying structure.
Outside these conditions, the diagnostics remain probe-dependent, in that we do not assume that an arbitrary restricted probe set retains all relevant information.
Nevertheless, our numerical results on an Ising benchmark show that the response entropy and spectrum track the qualitative changes in the corresponding process tensor MPO bond entropy and spectrum even when exact recovery is not imposed.
This provides a practical route to resolving temporal structure beyond a binary memory test, without reconstructing the complete multi-time process.

This work is organized as follows.
\Cref{sec:quantum-stochastic-processes} introduces relevant background about quantum stochastic processes, in particular, the process tensor framework and the causal break used to isolate memory carried outside the accessible system.
\Cref{sec:metric} constructs the response matrix directly from selected histories and future observations.
\Cref{sec:response-diagnostics} develops spectral diagnostics for the presence, structure, and persistence of memory, and formulates a theoretical search for accessible quantum-memory witnesses.
\Cref{sec:numerical-results} demonstrates two complementary uses of the framework where an interacting Ising process reveals how the response spectrum changes with coupling, temporal cut, and intervening resets, while a two-qubit storage-and-retrieval benchmark shows that memory can remain visible after it becomes classically reproducible.
An analytically known witness distinguishes the quantum and classical regimes of the latter benchmark.
Finally, \cref{sec:discussion} discusses the scope, limitations, and broader uses of the framework.
Supporting examples, numerical conventions, and consistency checks are collected in~\cref{app:classical_memory_entropy,app:quantum-memory-witness,app:numerical-consistency}.

\section{Quantum stochastic processes}
\label{sec:quantum-stochastic-processes}

At the most basic level, a stochastic process describes how observations made at different times are statistically related.
In the classical setting, one can assign random variables to a sequence of times and characterize the process by a consistent family of joint probability distributions.
In contrast, the quantum setting is more subtle because measurements and controls generally disturb the system.
Statistics at later times, therefore, depend both on the underlying dynamics and on the interventions applied at early times, so a sequence of system states does not fully describe the process.

\subsection{Process tensors and multi-time interventions}
Quantum multi-time statistics require an object that maps each causally ordered sequence of interventions to the probability of its record. The process tensor, equivalently a quantum comb, is that object \cite{PollockPRA2018,Milz2021,TarantoPRXQ2026}.
Once we have that, memory can be isolated by asking whether earlier interventions still affect the future after the accessible system has been discarded and prepared anew.

Consider an accessible system $S$ interacting with an environment $E$, with a (potentially correlated) initial state $\rho_0^{SE}$.
The experimenter intervenes on the system at times $t_0,\ldots,t_k$ and performs a terminal readout at $t_{\mathrm{out}}$.
At time $t_j$, the outcome $x_j$ is associated with a completely positive, trace-nonincreasing map $\mathcal A_j^{(x_j)}$.
The maps associated with all outcomes at a given time form an instrument when their sum is trace-preserving.
Between adjacent interventions, the completely positive, trace-preserving joint maps $\{\mathcal L_{j:j-1}\}_{j=1}^{k}$ propagate the system and environment.
A final joint map $\mathcal L_{\mathrm{out}:k}$ propagates them from the last intervention to the readout, where $\mathcal M_{\mathrm{out}}^{(r)}$ records the terminal outcome $r$.
Identity maps on the environment are suppressed for the system interventions and readout.

One intervention run produces the record $\boldsymbol\xi=(x_0,\ldots,x_k,r)$.
For the time-local interventions considered here, $\mathcal J_{\mathrm{out}:0}$ denotes the corresponding multi-time instrument, and $\mathbf A_{\mathrm{out}:0}^{(\boldsymbol\xi)}$ denotes the tensor product of the local Choi operators associated with that record \cite{TarantoPRXQ2026}.
We use input-before-output ordering for each Choi operator in a fixed product basis, where every transpose below is taken in that basis.

The initial state and all uncontrolled joint maps, including $\mathcal L_{\mathrm{out}:k}$, define the process tensor $\UpsilonPT_{\mathrm{out}:0}$ after the environment is traced out.
The process tensor maps the chosen intervention sequence to the probability of its observed record.
In the dynamical and Choi representations, respectively, this probability is
\begin{gather}
\begin{split}
 \Pr(\boldsymbol\xi\mid\mathcal J_{\mathrm{out}:0})
 =& \Tr\!\Big\{\big[
 \mathcal M_{\mathrm{out}}^{(r)}
 \circ \mathcal L_{\mathrm{out}:k}
 \circ \mathcal A_k^{(x_k)}
 \circ \mathcal{L}_{k:k-1}
 \\
 &\qquad \circ \mathcal A_{k-1}^{(x_{k-1})}
 \circ \cdots
 \\
 &\qquad \circ \mathcal{L}_{2:1}
 \circ \mathcal A_1^{(x_1)}
 \circ \mathcal{L}_{1:0}
 \circ \mathcal{A}_{0}^{(x_0)}
 \big](\rho^{SE}_0)\Big\} \\
=& 
 \Tr\!\left[
   \mathbf A_{\mathrm{out}:0}^{(\boldsymbol\xi)}
   \UpsilonPT_{\mathrm{out}:0}^{\mathsf T}
 \right].
 \label{eq:generalized-born-rule}
\end{split}
\end{gather}
The first expression follows the experiment in chronological order, while the second is the generalized Born rule in Choi form.
In this division, $\UpsilonPT_{\mathrm{out}:0}$ contains the initial state and uncontrolled evolution, while $\mathbf A_{\mathrm{out}:0}^{(\boldsymbol\xi)}$ contains the interventions and terminal readout.
Since the generalized Born rule holds for every admissible sequence, the process tensor contains the statistics of every admissible sequence of interventions.

\subsection{Causal breaks and environmental memory}

To isolate memory across a particular temporal cut, one asks whether resetting the accessible system removes every dependence on its earlier history.
A \textit{causal break} measures the incoming system, records the outcome, discards the system, and prepares a known state.
The preparation and subsequent controls are chosen independently of the earlier controls and break outcome.
Under these conditions, the accessible system cannot itself carry information from before the break, so any remaining dependence on the earlier controls or break outcome must be mediated by degrees of freedom that were not reset \cite{PollockPRL2018}.

More precisely, at the chosen cut time $t_c$, $\{E_c^{(i)}\}_i$ forms a POVM and $\sigma_c^{(p)}$ is a normalized state prepared after the measurement such that the causal break is the measure--prepare map
\begin{equation}
    \mathcal B_c^{(i,p)}(\rho)
    := \Tr\!\left(E_c^{(i)}\rho\right)\sigma_c^{(p)}.
    \label{eq:causal-break-map}
\end{equation}
where the trace gives the probability weight of the recorded outcome, while the output state is fixed by the preparation label $p$ rather than by the incoming system state.

In the dynamical expression in~\cref{eq:generalized-born-rule}, the causal break map then replaces the ordinary intervention $\mathcal A_c$ at $t_c$.
For an interior cut $c<k$, the resulting composition is
\begin{gather}
\begin{split}
&\underbrace{
\mathcal M_{\mathrm{out}}^{(r)}
\circ \mathcal L_{\mathrm{out}:k}
\circ \mathcal A_k^{(x_k)}
\circ \cdots
\circ \mathcal A_{c+1}^{(x_{c+1})}
\circ \mathcal L_{c+1:c}
}_{\text{post-break sequence}}
\circ
\underbrace{\mathcal B_c^{(i,p)}}_{\text{break at }t_c}
\\
&\qquad\circ
\underbrace{
\mathcal L_{c:c-1}
\circ \mathcal A_{c-1}^{(x_{c-1})}
\circ \cdots
\circ \mathcal L_{1:0}
\circ \mathcal A_0^{(x_0)}
}_{\text{pre-break sequence}}.
\label{eq:causal-break-sequence}
\end{split}
\end{gather}
such that the composition acts from right to left.
Thus, the interventions and outcomes through $t_{c-1}$, together with the cut outcome $i$, form the pre-break history.
The preparation $\sigma_c^{(p)}$, the later interventions, and the terminal readout form the future.
At the terminal cut $c=k$, the later controlled interventions are absent, but $\mathcal L_{\mathrm{out}:k}$ and the terminal readout still follow the break.

In the memoryless case, let $\mathbf h=(x_0,\ldots,x_{c-1})$ denote a pre-break history and let $\mathbf f=(x_{c+1},\ldots,x_k,r)$ denote the future outcome record.
Relative to the chosen history and future instruments, a process is conditionally memoryless at $t_c$ when
\begin{equation}
    \Pr(\mathbf f\mid \mathbf h,i,p)
    =\Pr(\mathbf f\mid p)
    \quad
    \text{for every $\mathbf h$ with $\Pr(\mathbf h,i)>0$}.
    \label{eq:conditional-memorylessness}
\end{equation}
Thus, once the preparation $p$ is fixed, the normalized future statistics do not depend on either the earlier history or the recorded break outcome.
A process is globally Markov when~\cref{eq:conditional-memorylessness} holds at every relevant cut and for all admissible instruments \cite{PollockPRL2018,PollockPRA2018}.

The process tensor contains the complete set of these intervention-dependent correlations, but its size grows rapidly with the number of interventions.
Process tensor matrix product operators (PT-MPOs) can compress this multi-time object when its temporal correlations are structured \cite{PollockPRA2018,Cygorek2025}.
Their bond spaces characterize the chosen representation of the full process, whereas a restricted intervention set may expose only part of that structure.
This leads us to the specific question: what properties of the process can be inferred directly from selected interventions and observations without reconstructing the complete process tensor?

\begin{figure*}[]
    \centering
    \includegraphics[width=0.7\textwidth]{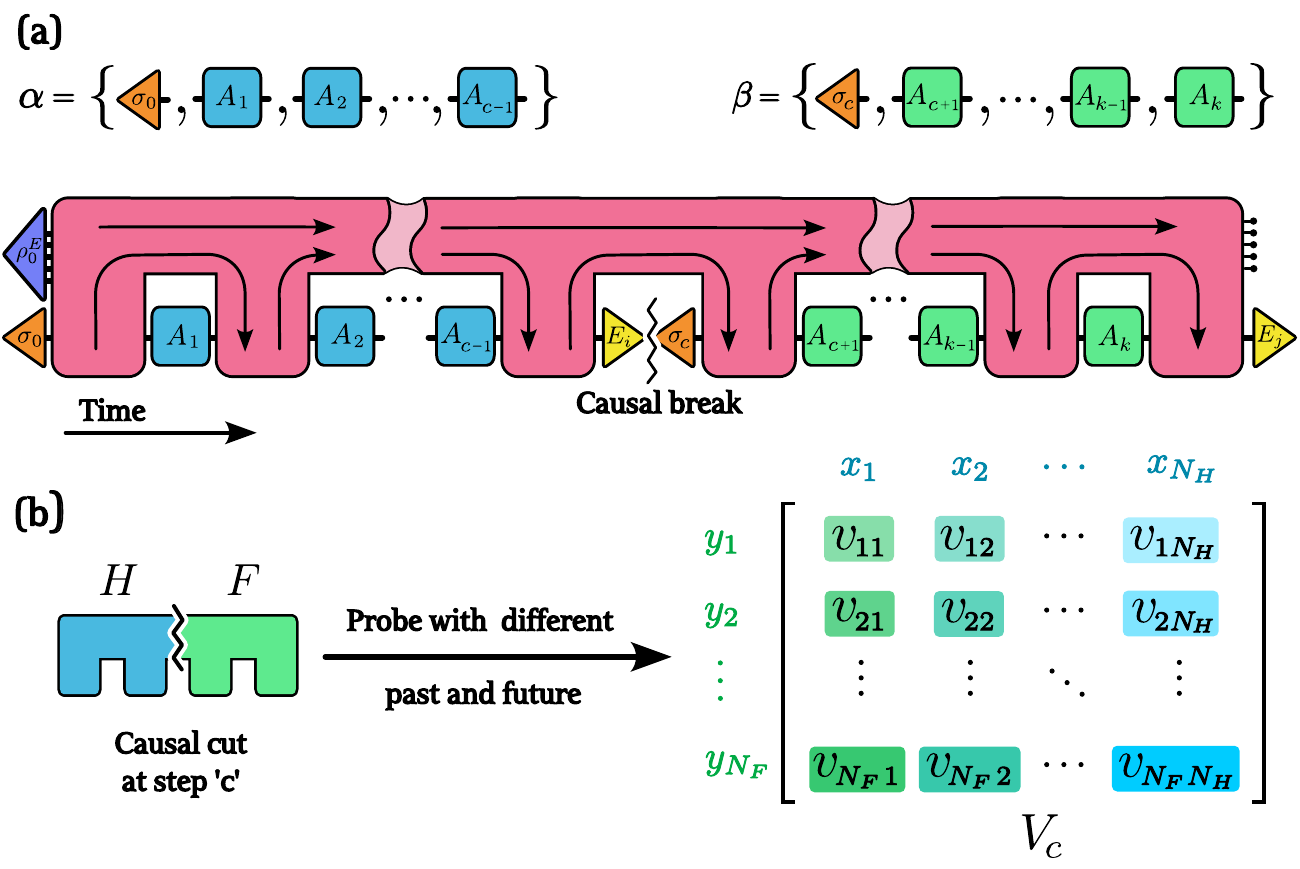}
    \caption{
        \textbf{Response matrix construction using an operational query across a causal break.} \textbf{(a)} The uncontrolled process begins in a general joint state $\rho_0^{SE}$, which need not factorize.
        The initial system intervention prepares the displayed state $\sigma_0$, and the selected history sequence $\alpha$ terminates with the physical measurement effect $E_c^{(i)}$.
        The protocol records outcome $i$, discards the incoming system, and prepares the state $\sigma_c^{(p)}$ independently of the history.
        It then applies a selected future sequence $\beta$, which terminates with the physical effect $E_{\mathrm{out}}^{(j)}$.
        \textbf{(b)} Repeating this causal break experiment for the selected history records $\mathbf x=(\alpha,i)$ and future records $\mathbf y=(p,\beta,j)$ produces the coefficients $v_{\mathbf x\mathbf y}$.
        The displayed probabilities are the POVM-response specialization of these coefficients.
        We arrange future records as rows and conditioned histories as columns, so the panel follows the convention in~\cref{eq:response-matrix-definition} and depicts $V_c$ directly. We thus can control the dimensions of the response matrix.
    }
    \label{fig:black_box}
\end{figure*}

\section{Response matrix construction}
\label{sec:metric}
We answer this question with a response matrix constructed directly from selected causal break experiments.
By placing conditioned histories and future probes on separate axes, the matrix captures the past--future dependence accessible to the chosen interventions and observations without reconstructing the full process tensor.

\subsection{From detecting memory to resolving its structure}
The causal break condition in~\cref{eq:conditional-memorylessness} provides an operational test for memory across a fixed temporal cut.
If two histories produce different normalized future statistics under the same break and future controls, then information from before the break remains accessible after it.
This distinction is sufficient to certify memory, but it does not describe the structure of that dependence.
A single violation does not tell us whether the tested future distinguishes one feature of the past or several independent features, nor how the observed response is distributed among those directions.
Comparing memory across temporal cuts and intervention sets therefore requires resolving these distinct directions of past--future dependence, not merely detecting that such dependence exists.

The complete process tensor contains this information in principle, but reconstructing it determines far more than is needed for this operational question and becomes expensive as the number of intervention times grows.
In a black-box experiment or simulation, one often has access only to a finite family of histories, preparations, future controls, and readouts.
The relevant object is then not the full latent process, but the dependence that those selected operations can resolve.
An operational description of this dependence must therefore separate history from future, preserve the probability weight of each conditioned history, and delimit exactly what the chosen probes can resolve.

Here, we introduce the response matrix as a sampled, operational representation of the past--future dependence accessible across a causal break.
We place each directly measured response at the intersection of one conditioned history and one future probe, so the matrix resolves the independent directions through which changes in the tested past appear in accessible future observations.
This construction requires neither process tensor reconstruction nor microscopic access to the system--environment dynamics.
We first define its entries through the causal break protocol and then derive how the matrix represents and normalizes the process information visible to the selected tester frames.

\subsection{Building the response matrix from causal breaks}
At intervention time $t_c$, we combine the pre-break sequence $\alpha$ and causal break outcome $i$ into the history record $\mathbf x=(\alpha,i)$.
We combine the independently chosen preparation $p$, later sequence $\beta$, and terminal effect or observable $j$ into the future record $\mathbf y=(p,\beta,j)$.
Thus $\mathbf x$ labels a conditioned past and $\mathbf y$ a selected future probe.
For each pair $(\mathbf x,\mathbf y)$, the protocol in~\cref{fig:black_box} executes $\alpha$, the causal break map $\mathcal B_c^{(i,p)}$, $\beta$, and terminal readout $j$ as one black-box query.
Once the settings have been selected, no query depends on the result of another.
The corresponding entries can therefore be evaluated independently and in parallel, then assembled into $V_c$. Experimentally, we assume reproducible sampling of the same process, including its initial joint state.

With the Choi ordering fixed above, the causal break map in~\cref{eq:causal-break-map} has Choi operator
\begin{equation}
    B_c^{(i,p)}
    = \left(E_c^{(i)}\right)^{\mathsf T}
      \otimes \sigma_c^{(p)},
    \label{eq:causal-break-choi}
\end{equation}
such that summing over $i$ yields a CPTP map. Let $A_j^{(x_j)}$ denote the Choi operator of a general completely positive intervention element $\mathcal A_j^{(x_j)}$. The labels $x_j, y_j$ specify both the chosen instrument setting and any retained outcome when applicable.
The distinct terminal effect $E_{\mathrm{out}}^{(j)}$ represents the map $\rho\mapsto\Tr(E_{\mathrm{out}}^{(j)}\rho)$ and contributes the Choi factor $(E_{\mathrm{out}}^{(j)})^{\mathsf T}$.
To allow a general, possibly correlated, pre-intervention state $\rho_0^{SE}$, we take the initial history element to be the replacement preparation $\mathcal P_0^{(a)}(\rho)=\Tr(\rho)\sigma_0^{(a)}$, whose Choi operator is
\begin{equation}
    A_0^{(a)}
    =\mathds{1}_{0^{\mathrm{in}}}\otimes\sigma_0^{(a)}.
\end{equation}
The history and future tester frames are the possibly redundant families $\{h^{(\mathbf x)}\}$ and $\{f^{(\mathbf y)}\}$ obtained by tensoring their local Choi factors.
Splitting the causal break Choi operator in~\cref{eq:causal-break-choi} across the history--future cut gives
\begin{gather}
\begin{split}
h^{(\mathbf x)}
&:= A_0^{(a)}
 \otimes A_1^{(x_1)} \otimes \cdots
 \otimes A_{c-1}^{(x_{c-1})}
 \otimes \left(E_c^{(i)}\right)^{\mathsf T},\\
f^{(\mathbf y)}
&:= \sigma_c^{(p)}
 \otimes A_{c+1}^{(y_{c+1})} \otimes \cdots
 \otimes A_k^{(y_k)}
 \otimes \left(E_{\mathrm{out}}^{(j)}\right)^{\mathsf T},
\end{split}
\end{gather}
Here $\alpha=(a,x_1,\ldots,x_{c-1})$ includes the initial preparation label $a$, while $\beta=(y_{c+1},\ldots,y_k)$.
At the boundary,
\begin{equation}
\begin{aligned}
    h^{(\mathbf x)}
    &=A_0^{(a)}\otimes(E_1^{(i)})^{\mathsf T}
    &&(c=1),\\
    f^{(\mathbf y)}
    &=\sigma_k^{(p)}\otimes(E_{\mathrm{out}}^{(j)})^{\mathsf T}
    &&(c=k).
\end{aligned}
\end{equation}
The final map $\mathcal L_{\mathrm{out}:k}$ remains inside the process tensor in the $c=k$ expression, between the preparation $\sigma_k^{(p)}$ and the terminal effect $E_{\mathrm{out}}^{(j)}$.
When the local elements belong to normalized instruments, their tensor products are tester elements in the usual sense.

We group all pre-break legs and the input side of the causal break slot into $H$, and the output side of that slot and all later legs into $F$.
We denote the process Choi operator in this grouping by $\UpsilonPT_{HF}$.
For a history record $\mathbf x$ and future record $\mathbf y$, the generalized Born rule defines the sampled response coefficient
\begin{gather}
v_{\mathbf{x}\mathbf{y}} :=
 \Tr\!\left[
    h^{(\mathbf x)} \otimes f^{(\mathbf y)}
   \UpsilonPT_{HF}^{\mathsf T}
 \right].
 \label{eq:born-rule}
\end{gather}
For CP intervention elements and physical POVM effects, $v_{\mathbf x\mathbf y}$ is the joint probability of all retained outcomes in the history and future records.
We can instead express a Hermitian tomographic observable as a real linear combination of measurement effects.
The same contraction then gives a history-probability-weighted expectation value and may be signed.

We collect these coefficients into the response matrix
\begin{equation}
    (V_c)_{\mathbf y\mathbf x}:=v_{\mathbf x\mathbf y},
    \label{eq:response-matrix-definition}
\end{equation}
whose rows label future records and whose columns label conditioned histories.
Thus each column describes the accessible future response associated with one selected history.
Panel~(b) of \cref{fig:black_box} uses the same orientation, with future records as rows and conditioned histories as columns, and therefore shows $V_c$ directly, completing its operational construction.

\subsection{Relating the response matrix to the underlying process}

Having constructed $V_c$ from measured responses, we now reorganize the same coefficients to relate the matrix to the underlying process and identify the direction that normalizes its columns.
Vectorization gives the shortest route to this relation.
Let $\lvert A\rangle\!\rangle:=\operatorname{vec}(A)$ denote column vectorization and let $\langle\!\langle A\rvert:=\operatorname{vec}(A)^\dagger$ denote its Hilbert--Schmidt dual.
We denote the vectorized history and future operator spaces by $\mathsf L_H$ and $\mathsf L_F$, with dimensions $D_H$ and $D_F$.
These Liouville spaces represent operators, not additional physical state spaces.
The transpose in the generalized Born rule gives the bilinear Choi pairing
\begin{equation}
    \Tr(A B^{\mathsf T})
    = \langle\!\langle A^*\vert B\rangle\!\rangle.
\end{equation}
Writing the matrix elements of $\UpsilonPT_{HF}$ in the fixed product basis, with $a,a'$ labeling history basis states and $b,b'$ labeling future basis states, we define
\begin{equation}
    \bigl(\Upsilon_{F\leftarrow H}\bigr)_{(b,b'),(a,a')}
    :=\bigl(\UpsilonPT_{HF}\bigr)_{(a,b),(a',b')}.
    \label{eq:process-reshaping}
\end{equation}
Thus the future index pair labels rows and the history index pair labels columns, with the first index varying fastest in each pair, as in column vectorization.
This regrouping leaves the Born-rule contraction unchanged, giving
\begin{gather}
    v_{\mathbf{x}\mathbf{y}}
    = \langle\!\langle (f^{(\mathbf y)})^* \vert
      \Upsilon_{F \leftarrow H}
      \vert h^{(\mathbf x)} \rangle\!\rangle.
\end{gather}
The matrix $\Upsilon_{F\leftarrow H}$ is not, in general, the Choi operator of a physical channel.
Let $N_H$ and $N_F$ be the numbers of sampled history and future records.
We define the corresponding frame matrices $H$ and $F$ as coordinate maps from $\mathsf L_H$ and $\mathsf L_F$ to the sampled record spaces:
\begin{equation}
    H := \sum_{\mathbf x} \lvert {\mathbf x}\rangle
        \langle\!\langle (h^{(\mathbf x)})^*\rvert
    \quad \mbox{and} \quad
    F := \sum_{\mathbf y} \lvert {\mathbf y}\rangle
        \langle\!\langle (f^{(\mathbf y)})^*\rvert.
\end{equation}
Thus $H\in\mathbb C^{N_H\times D_H}$, $F\in\mathbb C^{N_F\times D_F}$, and $\Upsilon_{F\leftarrow H}\in\mathbb C^{D_F\times D_H}$.
Collecting the vectorized Born-rule contractions for all sampled records yields the tester-frame factorization of the response matrix:
\begin{equation}\label{eq:respmat}
    V_c = F \Upsilon_{F \leftarrow H} H^{\mathsf{T}}.
\end{equation}
Here $V_c\in\mathbb C^{N_F\times N_H}$.
The ordinary transpose of $H$, rather than its adjoint, follows from the bilinear Choi pairing above.

Even when the conditional future response is independent of the history, columns of $V_c$ can differ in magnitude because the history branches have different probabilities.
The matrix $V_c$ retains its branch weights for the spectral analysis, where normalization is used to compare conditional future responses.
We now show how to recover the history-branch probability from the measured responses.
The future frame must span a deterministic tester $f^{(\mathrm{det})}$.
We construct it by holding the post-break preparation fixed, fixing a trace-preserving map or summing every outcome of a complete instrument at each later time, and summing every outcome of the terminal POVM.
We choose coefficients $n_{\mathbf y}$ that express this tester as
\begin{equation}
    f^{(\mathrm{det})}
    =\sum_{\mathbf y}n_{\mathbf y}f^{(\mathbf y)},
    \qquad
    \langle n\rvert
    :=\sum_{\mathbf y}n_{\mathbf y}\langle\mathbf y\rvert,
    \qquad n_{\mathbf y}\in\mathbb R.
\end{equation}
Then
\begin{equation}
    q_{\mathbf x}
    :=\langle n\rvert V_c\lvert\mathbf x\rangle
    \label{eq:history-branch-probability}
\end{equation}
is the probability of the corresponding history branch.
In an observable representation, the terminal normalization direction is $E_{\mathrm{out}}^{(0)}=\mathds{1}$, whose transpose is unchanged.
Although the coefficients $n_{\mathbf y}$ need not be unique for a redundant frame, every exact representation of the same deterministic tester gives the same $q_{\mathbf x}$.

\Cref{eq:respmat} also makes the scope of the construction explicit.
If the selected histories and future probes span the relevant operator spaces, $V_c$ determines the full response across the cut within those spaces.
Linearly independent probes give unique coordinates for this response, whereas redundant probes allow more than one coordinate representation.
If the probes do not span those spaces, $V_c$ determines only the part of the response that they can access.
With the construction and its scope established, we now use $V_c$ to identify and quantify the environmental memory visible to the chosen probes.

\section{What the Response Spectrum Reveals}
\label{sec:response-diagnostics}

The process tensor framework identifies environmental memory as history dependence that survives a causal break, while the response matrix $V_c$ collects selected instances of that dependence.
We now apply the standard singular-value decomposition (SVD)
\begin{equation}
    V_c
    =
    \sum_{m=1}^{r}
    \sigma_m
    \lvert\ell_m\rangle
    \langle r_m\rvert,
    \qquad
    \sigma_1\geq\cdots\geq\sigma_r>0,
    \label{eq:response-svd}
\end{equation}
where $\sigma_m$ are real, positive values and  $r=\operatorname{rank}(V_c)$.
We call the nonzero singular values,
\begin{equation}
    \operatorname{spec}_{\mathrm{resp}}(V_c)
    :=
    \{\sigma_m\}_{m=1}^{r},
    \label{eq:response-spectrum}
\end{equation}
the response spectrum and the paired singular directions are the response modes. 

A single causal break violation shows that the tested future retains some dependence on the past, but it does not reveal how many independent patterns carry that dependence or how the measured response is distributed among them.
Reconstructing the full process tensor would answer a much broader question and requires an informationally complete set of interventions \cite{WhitePRXQ2022}.
The response spectrum provides an intermediate description obtained directly from the selected experiments.
Its rank tells us whether the measured response remains confined to the one-direction memoryless baseline, while its singular vectors and relative weights resolve how the accessible response is organized.
Repeating the construction after a specified intervention tracks the persistence of those patterns.
The spectrum alone, however, may not be sufficient to determine whether their physical origin is classical or quantum.

We first establish the rank-one memoryless baseline, then determine what the remaining spectrum reveals about accessible memory modes, the unreconstructed process, and their persistence under interventions.
We finally identify the additional information needed to certify quantum memory.
Throughout, our conclusions concern only the memory visible to the chosen tester protocol.

Quantifying the non-Markovian memory via SVD across a past-future cut has appeared in numerical modelling contexts~\cite{PhysRevLett.120.240502, PhysRevA.106.022411, Guo2022, 10.3389/frqst.2025.1568471, VilkoviskiyPRB2024}. Our approach estimates the same quantity in a black-box setting.

\subsection{Detecting memory: the rank-one baseline}
\label{sec:rank-one-baseline}

To interpret the response spectrum as memory, we first identify its memoryless baseline.
If a causal break removes all tested dependence on the past, every history branch with nonzero probability produces the same normalized future response.
The columns of \(V_c\) then point in the same direction and differ only in their history probabilities.
We now show that the converse also holds, making this one-direction structure an exact test of conditional memorylessness within the chosen probes.

Let \(\langle n\rvert\) be the normalization functional associated with a deterministic future tester.
For an exact physical response, \(q_{\mathbf x}:=\langle n\rvert V_c\lvert\mathbf x\rangle\) is the branch probability of each history record \(\mathbf x\).
Whenever \(q_{\mathbf x}>0\), define the normalized future-response vector
\begin{equation}
    \lvert\phi_{\mathbf x}\rangle
    :=
    \frac{V_c\lvert\mathbf x\rangle}{q_{\mathbf x}}.
\end{equation}
Dividing by $q_{\mathbf x}$ removes the probability of obtaining that history, leaving the future response conditioned on it.
The tested process is conditionally memoryless at cut \(c\) when \(\lvert\phi_{\mathbf x}\rangle\) is independent of the nonzero-probability history record.
\begin{proposition}[Rank-one response and conditional memorylessness]
\label{prop:rank-one-response}
Let \(V_c\neq0\) be an exact response matrix satisfying \(q_{\mathbf x}=\langle n\rvert V_c\lvert\mathbf x\rangle\), and assume that \(V_c\lvert\mathbf x\rangle=0\) whenever \(q_{\mathbf x}=0\).
Then \(\operatorname{rank}(V_c)=1\) if and only if the normalized future response is the same for every tested history with nonzero probability.
\end{proposition}

\begin{proof}
If every history with nonzero probability gives the same normalized response \(\lvert\phi\rangle\), then \(V_c\lvert\mathbf x\rangle=q_{\mathbf x}\lvert\phi\rangle\).
By assumption, histories with zero probability give zero columns, and therefore
\begin{equation}
    V_c
    =
    \lvert\phi\rangle\langle q\rvert,
    \qquad
    \langle q\rvert
    :=
    \sum_{\mathbf x}q_{\mathbf x}\langle\mathbf x\rvert,
\end{equation}
which has rank one because \(V_c\neq0\).
Conversely, write a rank-one response as \(V_c=\lvert a\rangle\langle b\rvert\).
Since \(q_{\mathbf x}=\langle n\vert a\rangle\langle b\vert\mathbf x\rangle\), every history with \(q_{\mathbf x}>0\) satisfies
\begin{equation}
    \frac{V_c\lvert\mathbf x\rangle}{q_{\mathbf x}}
    =
    \frac{\lvert a\rangle}{\langle n\vert a\rangle},
\end{equation}
which is independent of \(\mathbf x\).
\end{proof}

Thus, rank greater than one certifies history dependence visible to the chosen controls and readouts.
Rank one establishes only tester-relative memorylessness unless the frames span the relevant Liouville spaces, that is, unless they are informationally complete.
With such frames, rank one implies process factorization at this cut, where global Markovianity additionally requires this condition at every relevant cut and for all admissible interventions.
Within these limits, the rank gives a direct memory test such that a response matrix with a rank greater than one detects memory, while rank one identifies the conditionally memoryless baseline for the chosen probes.
Thus, response rank greater than one subsumes the role of a memory witness with respect to the selected probes.

\subsection{Resolving memory: response modes and spectral complexity}
\label{sec:memory-spectrum-modes}

Having used the rank for a yes-or-no memory test, we turn to the full response spectrum to resolve the structure of the measured response.
When the rank exceeds one, we use the spectrum to determine how many independent response directions the chosen probes resolve and how the response is distributed among them.
In the SVD of~\cref{eq:response-svd}, each right singular vector selects a linear combination of tested histories, while the corresponding left singular vector gives the associated future response.
We call each such pair a response mode.
Because singular vectors combine multiple probes, individual modes need not represent physical interventions.

We use the SVD's optimal low-rank approximation \cite{Eckart1936} to quantify how much of the measured response each mode carries.
Retaining the \(k\) largest modes gives the best rank-\(k\) representation, while discarding the remaining modes incurs squared Frobenius error \(\sum_{m>k}\sigma_m^2\).
Thus \(\sigma_m^2\) gives the contribution of mode \(m\) to the total squared response.
For \(V_c\neq0\), we normalize these contributions to define the mode weights
\begin{equation}
    p_m
    :=
    \frac{\sigma_m^2}{\sum_j\sigma_j^2}
    =
    \frac{\sigma_m^2}{\lVert V_c\rVert_{\mathrm F}^2}.
    \label{eq:mode_weights}
\end{equation}
The weight \(p_m\) is the fraction of the squared Frobenius norm carried by mode \(m\).
We summarize this distribution by defining the response entropy and effective mode number:
\begin{equation}
    S_V
    :=
    -\sum_m p_m\log p_m,
    \qquad
    R
    :=
    \exp(S_V).
    \label{eq:SV_formal}
\end{equation}
We use \(S_V\) to measure how evenly the response is distributed across its modes, while \(R\) gives the corresponding effective mode number.
These definitions give \(0\leq S_V\leq\log r\) and \(1\leq R\leq r\) such that a rank-one response has \(S_V=0\) and \(R=1\), while \(R=r\) when all nonzero modes have equal weight.
More generally, \(R\) is the number of equally weighted modes that would give the same entropy, so a value near one means that one mode dominates.
Thus \(S_V\) and \(R\) describe how the response is distributed among modes, not its total size.
For a higher-rank response, the leading singular term is the best rank-one approximation, but it need not define a physical Markov process.

\Cref{app:classical_memory_entropy} illustrates this interpretation with a bath that stores a classical bit across a system reset.
When the probes resolve both bit values, $V_c$ has two nonzero modes, provided both values occur.
Equally likely values give $S_V=\ln2$, while a strongly biased bit gives a small entropy despite perfect recall.
Ignoring the history outcome leaves a single response column and $S_V=0$, although the bath dynamics are unchanged.
Thus, $S_V$ describes the distribution of response weight among the patterns resolved by the chosen probes.
Changing the probe set or its relative sampling and weighting can change this distribution even for a fixed process.
The numerical results below therefore characterize memory relative to their specified probe protocols.

We estimate the response spectrum directly from a finite set of causal break measurements.
Increasing the number of sampled history and future probes tests the convergence of \(S_V\) and \(R\), while repeated measurements improve the precision of \(V_c\).
We retain singular values above the numerical or statistical resolution as resolved modes and use consistent sampling and weighting across cuts to compare how the accessible memory structure changes throughout the process.

We next ask when the rank and spectrum measured from \(V_c\) equal those of the underlying process tensor.

\subsection{Inferring process structure without reconstruction}
\label{sec:pt-mpo-connection}

The rank of \(V_c\) describes the measured response.
Under suitable probe conditions, it can equal the corresponding rank of the process tensor, and its relative singular-value weights can equal those of a PT-MPO.
We now state the conditions for both results.

The two ranks apply to different matrices.
The response rank \(r=\operatorname{rank}(V_c)\) counts the independent response patterns distinguished by the chosen probes.
To define the corresponding quantity for the process tensor, we reshape it into the matrix \(\Upsilon_{F\leftarrow H}\).
The rank of this matrix is the operator-Schmidt rank, \(r_{\mathrm{OS}}^{(c)}:=\operatorname{rank}(\Upsilon_{F\leftarrow H}) =\operatorname{OSR}_{H|F}(\UpsilonPT_{HF})\).
It counts the independent history--future directions present in the process, including those that the chosen probes may not resolve.
The reshaped process has the singular-value decomposition
\begin{equation}
    \Upsilon_{F\leftarrow H}
    =
    \sum_{m=1}^{r_{\mathrm{OS}}^{(c)}} s_m
    \lvert a_m\rangle\!\rangle
    \langle\!\langle b_m\rvert,
    \qquad
    s_1\geq\cdots\geq s_{r_{\mathrm{OS}}^{(c)}}>0.
\end{equation}
The response has singular values \(\sigma_m\).
An exact PT-MPO in mixed-canonical form at the same temporal cut has nonnegative bond coefficients \(\lambda_m\).
These coefficients equal \(s_m\) when the PT-MPO represents the same normalized process and uses the same history--future division, index ordering, and vectorization as the response matrix.
We define the normalized PT-MPO bond weights and their entropy by
\begin{equation}
    \begin{aligned}
        p_m^{\mathrm{MPO}}
        &:=
        \frac{\lambda_m^2}{\sum_j\lambda_j^2},\\
        S_{\mathrm{MPO}}
        &:=
        -\sum_m p_m^{\mathrm{MPO}}\log p_m^{\mathrm{MPO}}.
    \end{aligned}
\end{equation}

The probe matrices can remove process directions or change their relative weights.
The following result states when they preserve the rank and spectrum.
\begin{proposition}[Recovery of process rank and spectrum]
\label{prop:response-process-recovery}
Let \(V_c=F\Upsilon_{F\leftarrow H}H^{\mathsf T}\), and let \(\chi_c\) be the bond dimension across cut \(c\) of an exact PT-MPO representation of the same process.
For arbitrary probe frames,
\begin{equation}\label{eq:response-process-rank}
    \begin{gathered}
        \operatorname{rank}(V_c)
        \leq
        \min\!\left\{
            \operatorname{rank}(F),
            \operatorname{rank}(H),
            r_{\mathrm{OS}}^{(c)}
        \right\},\\
        r_{\mathrm{OS}}^{(c)}\leq\chi_c.
    \end{gathered}
\end{equation}
If the probes span both operator spaces, then
\begin{equation}
    \operatorname{rank}(V_c)=r_{\mathrm{OS}}^{(c)}.
\end{equation}
Suppose further that positive constants \(\alpha\) and \(\beta\) satisfy
\begin{equation}
    F^\dagger F=\alpha\mathds{1}_{D_F},
    \qquad
    H^\dagger H=\beta\mathds{1}_{D_H}.
\end{equation}
The response and process singular values then obey
\begin{equation}\label{eq:response-process-singular-values}
    \sigma_m=\sqrt{\alpha\beta}\,s_m.
\end{equation}
Under the PT-MPO conventions above, this also gives
\begin{equation}
    p_m=p_m^{\mathrm{MPO}},
    \qquad
    S_V=S_{\mathrm{MPO}}.
\end{equation}
\end{proposition}

\begin{proof}
The rank of a matrix product cannot exceed the rank of any factor.
Applying this to \(V_c=F\Upsilon_{F\leftarrow H}H^{\mathsf T}\) gives the first bound in~\cref{eq:response-process-rank}.
The operator-Schmidt rank is the minimum bond dimension across the history--future division, which gives the second bound \cite{KeelingPRX2026,Cygorek2025}.

When the probes span both spaces, \(F\) has a left inverse \(L_F\), and \(H^{\mathsf T}\) has a right inverse \(R_H\).
Thus \(L_FF=\mathds{1}_{D_F}\) and \(H^{\mathsf T}R_H=\mathds{1}_{D_H}\).
Applying these inverses to \(V_c=F\Upsilon_{F\leftarrow H}H^{\mathsf T}\) gives \(\Upsilon_{F\leftarrow H}=L_FV_cR_H\).
Its rank cannot exceed that of \(V_c\), which combines with~\cref{eq:response-process-rank} to prove rank equality.

For the spectrum, set \(Q_F:=F/\sqrt{\alpha}\) and \(Q_H:=H/\sqrt{\beta}\).
Then \(Q_F^\dagger Q_F=\mathds{1}_{D_F}\) and \(Q_H^{\mathsf T}Q_H^*=\mathds{1}_{D_H}\), so \(Q_F\) has orthonormal columns and \(Q_H^{\mathsf T}\) has orthonormal rows.
Multiplication by these matrices preserves the nonzero singular values.
Substituting \(V_c=\sqrt{\alpha\beta}\,Q_F\Upsilon_{F\leftarrow H}Q_H^{\mathsf T}\) proves \cref{eq:response-process-singular-values}.
Finally, \(\lambda_m=s_m\), and the common factor in \(\sigma_m=\sqrt{\alpha\beta}\,s_m\) cancels when we normalize the squared singular values.
This gives the stated equalities for \(p_m\) and \(S_V\).
\end{proof}

The two assumptions play distinct roles.
Spanning probes make the response rank equal the operator-Schmidt rank.
Equal directional weighting also makes the normalized response spectrum equal the PT-MPO bond spectrum.
Incomplete probe spans can miss process directions, and unequal weighting can change their relative spectral weights.
The sampled frames used below are not assumed to satisfy the equal-directional-weighting condition, so their response and PT-MPO spectra need not coincide.
For a truncated PT-MPO, the rank and spectrum refer to the approximated process rather than the exact \(\Upsilon_{F\leftarrow H}\).
These recovery conditions concern a fixed history--future division.
They do not tell us whether the observed history dependence survives when we intervene before probing the future.

\subsection{Memory persistence under interventions}
\label{sec:memory-length}

In the temporal setting, memory persistence is not set by the size of the bond, as is the case for MPS. This is due to the fact that information can traverse across the system channel as well as the memory.
Memory persistence asks how long the past continues to influence the future under a specified sequence of interventions.
An instrument-specific memory length is the smallest number of intervening steps for which a chosen block instrument makes the future independent of the history for every outcome that can occur.
This exact condition may never be reached, but we can still track how the remaining dependence changes as the block grows.
Different interventions can preserve or erase different parts of the memory, so both questions depend on the chosen instrument~\cite{TarantoPRL2019,TarantoPRA2019,Taranto3}.

The response matrix makes this test operational.
We keep the selected history and future probes fixed, place an intervention block \(M\) between them, and condition on one outcome of that block.
The resulting matrix has the same history columns and future rows as before, so changes across increasing block lengths isolate the effect of the intervening control.
At each length, we test whether every nonzero outcome branch has reached the rank-one baseline.
Repeating this construction for longer blocks distinguishes complete memory removal from continued, but changing, history dependence.

Formally, let \(\mathcal J_M=\{\mathbf A_M^{(\mathbf z)}\}\) be an instrument on a block \(M\), whose elements sum to a deterministic multi-time operation.
Conditioning on outcome \(\mathbf z\) gives the process and response matrix
\begin{gather}\label{eq:respmatcond}
    \begin{split}
    \widetilde{\UpsilonPT}_{HF|\mathbf z,\mathcal J_M}^{\mathsf T}
    &:=
    \Tr_M\!\left[
        \mathbf A_M^{(\mathbf z)}
        \UpsilonPT_{HMF}^{\mathsf T}
    \right], \\
    V_{c|\mathbf z,\mathcal J_M}
    &:=
    F\widetilde{\Upsilon}_{F\leftarrow H|\mathbf z,\mathcal J_M}H^{\mathsf T}.
    \end{split}
\end{gather}
The instrument removes memory when every nonzero conditional process factorizes as
\begin{equation}
    \widetilde{\UpsilonPT}_{HF|\mathbf z,\mathcal J_M}^{\mathsf T}
    =
    \widetilde{\UpsilonPT}_{H|\mathbf z,\mathcal J_M}^{\mathsf T}
    \otimes
    \UpsilonPT_{F|\mathbf z,\mathcal J_M}^{\mathsf T}.
    \label{eq:conditional-factorization}
\end{equation}
We causally normalize the future factor, leaving the outcome weight in the history factor.
\begin{corollary}[Response criterion for instrument-specific memory length]
\label{cor:response-memory-length}
With informationally complete history and future frames, an instrument removes memory between \(H\) and \(F\) if and only if every nonzero outcome branch produces a rank-one conditional response matrix.
For a specified nested family of length-\(\ell\) blocks at a fixed position, the smallest \(\ell\) that satisfies this condition is the instrument-specific memory length at that position.
\end{corollary}

\begin{proof}
The factorization in~\cref{eq:conditional-factorization} is equivalent to a rank-one reshaping.
Informationally complete frames preserve this rank, so the conditional process factorizes if and only if its response matrix has rank one.
\end{proof}

For a nested family of blocks \(M_\ell\), denote the response for a chosen outcome record by \(V_\ell:=V_{c|\mathbf z_\ell,\mathcal J_{M_\ell}}\).
With the history and future frames fixed, a rank above one shows that tested history dependence survives the block, while the entropy \(S_V(\ell)\) tracks how the surviving response is distributed among its modes.
This sequence gives the persistence profile for the chosen intervention and outcome record.

With a finite set of frames, rank one guarantees independence only within the chosen probes.
A process-wide Markov order further requires the all-outcome condition at every relevant placement of the instrument.
The response profile therefore tells us how memory survives the chosen intervention, but not whether that memory is classical or quantum.
We address that distinction next.

\subsection{Certifying quantum memory}
\label{sec:quantum-memory-certification}

The response spectrum tells us whether the future retains a dependence on the tested history and how that dependence is distributed among the resolved response modes.
It does not, by itself, identify with certainty the physical nature of the memory that carries this dependence.
A classical environmental record can produce several response modes and a nonzero response entropy, just as a coherent quantum memory can.
Appendix~\ref{app:classical_memory_entropy} illustrates this with a single bath qubit where information passing across the causal break can be classical, and still give a non-trivial response spectrum.
We can therefore ask a stronger question: can the response matrix provide us a witness that can be used to witness quantum memory in our process being probed?

Let $\Upsilon_{HF}\geq 0$ denote the Choi operator of the process, with all physical Choi legs before the selected causal break grouped into $H$ and all legs after it grouped into $F$.
The measure--prepare operation at the causal break has the product Choi form $E_i^{T}\otimes\sigma_p$ and hence does not itself carry quantum correlations from $H$ to $F$.
Any correlation detected across this division is therefore attributed to the process sampled by the chosen history and future probes.
This temporal partition is essential because splitting the input and output legs of an ordinary channel can reveal channel Choi entanglement even for a memoryless process~\cite{MilzSciPost2021}.

To make ``classical memory'' explicit, we follow the separability across the selected temporal cut as the necessary condition for a classical feed-forward memory model~\cite{Peres1996Separability,Giarmatzi2021},
\begin{equation}
    \Upsilon_{HF}^{\mathrm{cl}} = \sum_{\lambda} R_\lambda^H \otimes S_\lambda^F, \qquad R_\lambda^H, S_\lambda^F \succeq 0.
    \label{eq:classical-memory-separable}
\end{equation}
Here, the classical label $\lambda$ carries information from the history to the future, while the conditional operators on either side of the cut can remain quantum.
Different operational classical memory, mixtures of memoryless processes, and separable process tensors can define classes beyond two times~\cite{Taranto2024}, while alternative notions based on Kolmogorov-consistent sequential statistics have also been discussed in previous literature~\cite{Strasberg2019,MilzPRX2020}.

We therefore search for a witness nonnegative on the cut-separable set in Eq~\eqref{eq:classical-memory-separable}, which contains the classical memory processes considered here.
A negative witness value for the response matrix then rules out any classical feed-forward decomposition across \(H:F\), thereby certifying quantum memory accessible to the chosen instruments.
If no such negativity is found, the result remains inconclusive, since the process may contain quantum memory that the selected probes or witness fail to detect.

We can establish entanglement across \(H:F\) using the positive-map criterion.
Let \(\Lambda_F:\mathcal L(\mathcal H_F)\to\mathcal L(\mathcal H_{F'})\) be a map that sends positive operators to positive operators.
For every separable operator in Eq.~\eqref{eq:classical-memory-separable},
\begin{equation}
  (\mathds{1}_H\otimes\Lambda_F)
  (\Upsilon_{HF}^{\mathrm{cl}})
  =\sum_{\lambda}R_{\lambda}^{H}\otimes
  \Lambda_F(S_{\lambda}^{F})\geq0 .
  \label{eq:positive-map-on-classical-process}
\end{equation}
If a vector \(\lvert\nu\rangle\in\mathcal H_H\otimes\mathcal H_{F'}\) satisfies
\begin{equation}
    \langle\nu\rvert
    (\mathds{1}_H\otimes\Lambda_F)(\UpsilonPT_{HF})
    \lvert\nu\rangle <0,
    \label{eq:positive-map-negative-direction}
\end{equation}
then $\Upsilon_{HF}$ cannot have the form Eq.~\eqref{eq:classical-memory-separable}, and so, is entangled across the cut~\cite{Horodecki1996}.

Defining $Q:=|\nu\rangle\langle\nu|\geq0$ and using the Hilbert--Schmidt adjoint gives the witness
\begin{align}
  W&=(\mathds{1}_H\otimes\Lambda_F^{\dagger})(Q),
  \\
  \operatorname{Tr}(W\Upsilon_{HF})
  &=\operatorname{Tr}\!\left[
  Q(\mathds{1}_H\otimes\Lambda_F)(\Upsilon_{HF})\right].
  \label{eq:positive-map-witness}
\end{align}
The right-hand side is nonnegative for every separable process because it is the trace of a product of two positive operators, whereas Eq.~\eqref{eq:positive-map-negative-direction} makes it negative for the tested process, violating the assumption in Eq~\eqref{eq:classical-memory-separable}.

The response matrix does not provide arbitrary expectation values of the full process.
It provides the fixed probe contractions
\begin{equation}
  B_{xy}:=\bigl[h^{(x)}\bigr]^T\otimes\bigl[f^{(y)}\bigr]^T,
  \qquad
  (V_c)_{yx}=\operatorname{Tr}(B_{xy}\Upsilon_{HF}),
  \label{eq:response-probe-operators}
\end{equation}
and therefore accesses the real operator space $\mathcal S_{\mathrm{acc}}:=\operatorname{span}_{\mathbb R}\{B_{xy}\}$.
A valid negative witness in this subspace certifies quantum memory in the full process without reconstructing $\Upsilon_{HF}$.
Failure to find one is inconclusive: the process may be separable, the chosen positive map may miss its entanglement, or the selected probes may not resolve the required witness.
Conversely, projecting an inaccessible witness into $\mathcal S_{\mathrm{acc}}$, or partially transposing a pseudoinverse reconstruction without independently validating it, need not preserve the witness inequalities and can lead to a false certificate.

\begin{proposition}[Response-matrix quantum-memory witness]
\label{prop:response-quantum-memory-witness}
Let the history and future probes be Hermitian, and let \(c_{\mathbf x\mathbf y}\in\mathbb R\).
Define the accessible Hermitian operator
\begin{equation}
    W_C
    :=
    \sum_{\mathbf x,\mathbf y}
    c_{\mathbf x\mathbf y}
    \bigl[h^{(\mathbf x)}\bigr]^{\mathsf T}
    \otimes
    \bigl[f^{(\mathbf y)}\bigr]^{\mathsf T}.
    \label{eq:response-space-witness}
\end{equation}
Suppose that \(W_C\) is nonnegative on every positive operator $R_\lambda^H,S_\lambda^F\succeq0$ separable across the temporal cut \(H:F\), i.e.,
\begin{equation}
    \operatorname{Tr}(W_C\Sigma_{HF})\geq0
    \quad
    \forall \quad
    \Sigma_{HF}
    =
    \sum_{\lambda}R_\lambda^H\otimes S_\lambda^F.
    \label{eq:response-witness-separable-bound}
\end{equation}
Then its expectation value on the process is obtained directly from the response matrix:
\begin{equation}
    \operatorname{Tr}(W_C\UpsilonPT_{HF})
    =
    \sum_{\mathbf x,\mathbf y}
    c_{\mathbf x\mathbf y}(V_c)_{\mathbf y\mathbf x}.
    \label{eq:response-witness-contraction}
\end{equation}
Consequently,
\begin{equation}
    \operatorname{Tr}(W_C\UpsilonPT_{HF})
    \begin{cases}
        <0,
        & \text{certifies quantum memory across \(H:F\)},\\[2pt]
        \geq0,
        & \text{is inconclusive for this witness}.
    \end{cases}
    \label{eq:response-witness-criterion}
\end{equation}
\end{proposition}

\begin{proof}
For each pair of probes, the process Born rule and the identity \(\operatorname{Tr}(AB^{\mathsf T}) =\operatorname{Tr}(A^{\mathsf T}B)\) give
\begin{equation*}
    (V_c)_{\mathbf y\mathbf x}
    =
    \operatorname{Tr}\!\left[
        \left(
        \bigl[h^{(\mathbf x)}\bigr]^{\mathsf T}
        \otimes
        \bigl[f^{(\mathbf y)}\bigr]^{\mathsf T}
        \right)
        \UpsilonPT_{HF}
    \right].
\end{equation*}
Multiplying these equalities by the real coefficients \(c_{\mathbf x\mathbf y}\), summing over the probe labels, and using the linearity of the trace yields
\begin{align*}
    \sum_{\mathbf x,\mathbf y}
    c_{\mathbf x\mathbf y}(V_c)_{\mathbf y\mathbf x}
    &=
    \operatorname{Tr}\!\left[
        \left(
        \sum_{\mathbf x,\mathbf y}
        c_{\mathbf x\mathbf y}
        \bigl[h^{(\mathbf x)}\bigr]^{\mathsf T}
        \otimes
        \bigl[f^{(\mathbf y)}\bigr]^{\mathsf T}
        \right)
        \UpsilonPT_{HF}
    \right]
    \nonumber\\
    &=
    \operatorname{Tr}(W_C\UpsilonPT_{HF}).
\end{align*}

Now assume that the process admits the classical memory decomposition in \cref{eq:classical-memory-separable}.
By the defining property in \cref{eq:response-witness-separable-bound}, it would necessarily satisfy \(\operatorname{Tr}(W_C\UpsilonPT_{HF})\geq0\).
Therefore, an observed negative value contradicts every such decomposition and certifies entanglement across the selected temporal cut.
Since a single witness need not detect every entangled operator, a nonnegative value does not establish separability and remains inconclusive.
\end{proof}

A witness satisfying the proposition can be searched for directly within the measured probe span using semidefinite programs.
This formulation assumes exact, consistent response data, while treating statistical uncertainty is left to future work.
Let \(\Lambda_F\) be a fixed positive map acting on the future space, with Hilbert--Schmidt adjoint \(\Lambda_F^\dagger\).
Consider the semidefinite program
\begin{equation}
\begin{aligned}
  \underset{\{c_{\mathbf x\mathbf y}\},P,Q}{\operatorname{minimize}}
  \quad&
  \sum_{\mathbf x,\mathbf y}
  c_{\mathbf x\mathbf y}(V_c)_{\mathbf y\mathbf x}
  \\[2pt]
  \text{subject to}\quad&
  \sum_{\mathbf x,\mathbf y}
  c_{\mathbf x\mathbf y}B_{\mathbf x\mathbf y}
  =
  P+
  (\mathds{1}_H\otimes\Lambda_F^\dagger)(Q),
  \\
  &P\succeq0,\qquad Q\succeq0,
  \qquad c_{\mathbf x\mathbf y}\in\mathbb R,
  \\
  &\operatorname{Tr}P+\operatorname{Tr}Q\leq1.
\end{aligned}
\label{eq:response-witness-sdp}
\end{equation}

The operator on the left-hand side is precisely \(W_C\).
Hence the equality constraint simultaneously restricts the witness to \(\mathcal S_{\mathrm{acc}}\) and imposes a more general positive-map form
\begin{equation}
    W_C=P+(\mathds{1}_H\otimes\Lambda_F^\dagger)(Q).
    \label{eq:positive-map-witness-form}
\end{equation}

For any separable positive operator \(\Upsilon_{HF}^{\mathrm{cl}}\), positivity of \(\Lambda_F\) guarantees that \((\mathds{1}_H\otimes\Lambda_F)\Upsilon_{HF}^{\mathrm{cl}}\succeq0\).
Since \(P\) and \(Q\) are positive operators, the adjoint identity implies \(\operatorname{Tr}(W_C\Upsilon_{HF}^{\mathrm{cl}})\geq0\) for every feasible witness in \cref{eq:response-witness-sdp}.

The trace constraint fixes the otherwise arbitrary witness scale.
Indeed, if a feasible witness had a negative expectation value, multiplying \(P,Q\), and all \(c_{\mathbf x\mathbf y}\) by an arbitrarily large positive number would preserve the other constraints while driving the objective to \(-\infty\).
Since \(P, Q\succeq0\), their traces equal their trace norms, so \(\operatorname{Tr}P+\operatorname{Tr}Q\leq1\) provides a convenient convex normalization.
The inequality, rather than an equality, also keeps the trivial choice
\begin{equation}
    P=Q=0,\qquad c_{\mathbf x\mathbf y}=0
\end{equation}
feasible, giving objective value zero.

This makes the SDP output particularly transparent.
If \(\UpsilonPT_{HF}\) is separable across the chosen cut, every feasible witness has a nonnegative objective, where since the zero witness is feasible, the minimum must be exactly zero.
Conversely, a certified optimum
\begin{equation}
    \sum_{\mathbf x,\mathbf y}
    c_{\mathbf x\mathbf y}^{\star}
    (V_c)_{\mathbf y\mathbf x}<0
\end{equation}
provides the accessible witness \(W_C^\star=\sum_{\mathbf x,\mathbf y} c_{\mathbf x\mathbf y}^{\star}B_{\mathbf x\mathbf y}\) and rules out every classical decomposition in \cref{eq:classical-memory-separable}.
An optimum equal to zero only means that no certificate was found for the chosen probes and positive map, such that it does not establish that the underlying process memory is classical.
The partial-transpose specialization and further practical qualifications are given in Appendix~\ref{app:quantum-memory-witness}.

The response matrix therefore turns any independently valid witness in the sampled probe span into a weighted combination of measured responses.
Together, its rank, spectrum, intervention dependence, and witness value reveal whether memory remains, how it is organized, how long it persists, and whether it is incompatible with the adopted classical model.
We now test these uses in an interacting-spin process and a two-qubit memory benchmark.

There's no limitation in generalising both the response matrix and quantum memory witness to a multi-time setting to enable witnessing genuine multipartite entanglement in memory~\cite{MilzSciPost2021, White2025whatcanunitary}.

\section{Numerical results}
\label{sec:numerical-results}
We now apply the diagnostics developed above to controlled numerical processes.
First, we use an interacting-spin process to study how history dependence appears in the response spectrum, how probe choice shapes that spectrum, and how the response changes under intervening operations.
We then turn to a separate two-qubit process and test whether its memory can be certified as quantum.

For the interacting-spin process, we use an open transverse-field Ising chain of \(L\) spins, with nearest-neighbor coupling \(J\) and transverse field \(g\), governed by
\begin{equation}
    H_{\mathrm{spin}}
    =-J\sum_{r=0}^{L-2}Z_rZ_{r+1}-g\sum_{r=0}^{L-1}X_r,
    \label{eq:numerical-ising-hamiltonian}
\end{equation}
where the first spin is the accessible system, and the remaining spins form the environment.
All controls and readouts act on the first spin.
For each test here, we set \(L=6\), initialize all spins in \(\ket{0}\), fix the transverse field at \(g=1\), and use evolution intervals of \(\Delta t=0.1\).
The chain evolves under this Hamiltonian for one interval before the first intervention, between interventions, and between the last intervention and the final readout.
We simulate these dynamics using exact statevector evolution.

We construct each numerical response matrix from the normalized final state of the accessible system by the expectation values of the identity and the three Pauli operators, ordered as \(I,X,Y,Z\).
For each history--future setting pair, we weight all four values by the joint probability of the measurement such that the weighted identity value therefore records this branch probability.
Each history contributes one column, and each future setting contributes four response rows, where we evaluate these quantities directly from the simulated states, without finite-shot sampling.
The response entropy \(S_V\) and effective mode number \(R\) use normalized squared singular values and natural logarithms.
In the Ising calculations, we discard only a spectral tail carrying at most \(10^{-12}\) of the total squared weight and renormalize the retained weights.
Each study below specifies its probe ensemble and intervention sequence.
All results here were generated using MQT-YAQS~\cite{Wille2024,YAQS} and performed on an Intel i5-13600KF processor.

\begin{figure}[]
\centering
\includegraphics[width=\columnwidth]{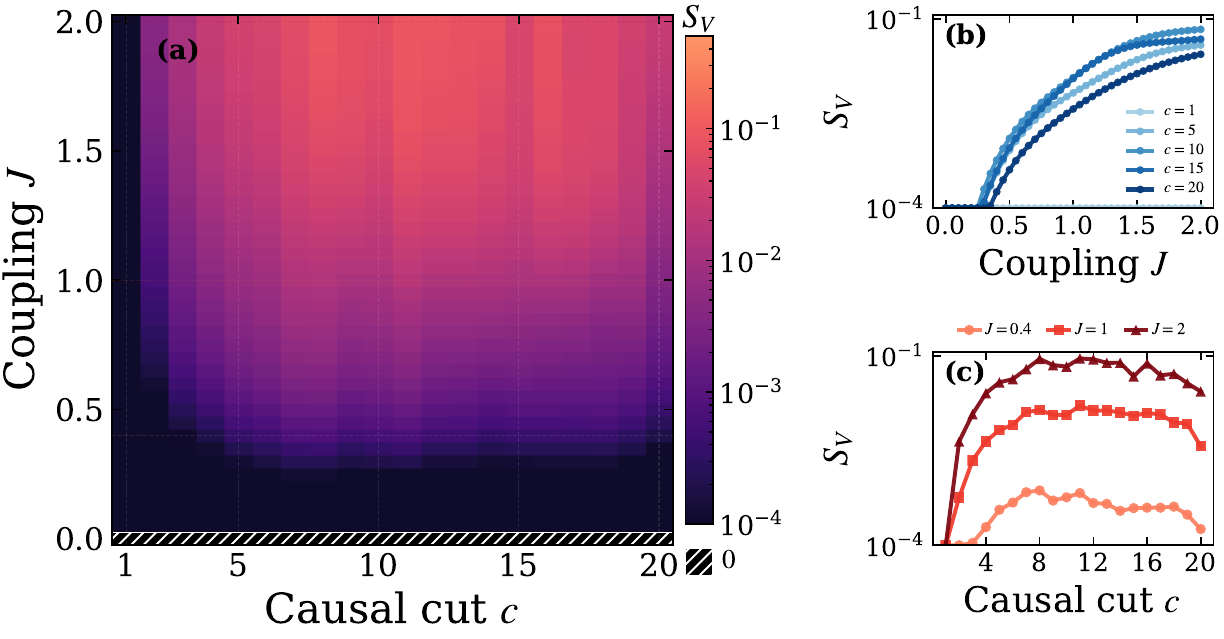}
    \caption{
        \textbf{Memory landscape.} \textbf{(a)} Response entropy \(S_V(c,J)\) for the transverse-field Ising benchmark as a function of causal cut \(c\) and interaction strength \(J\), with \(g=1\), \(L=6\), \(k=20\), \(\Delta t=0.1\), 64 conditioned history settings, and 64 future settings.
        \textbf{(b)} Slices at fixed cuts \(c=1,5,10,15,20\).
        \textbf{(c)} Slices at fixed couplings \(J=0.4,1,2\).
        For these probes, the entropy builds into an interior plateau and increases with coupling.
        Positive values below \(10^{-4}\) are clipped for visualization, and the black \(J=0\) row with white hatching denotes the exact \(S_V=0\) result.
        The landscape describes the response to the chosen history and future probes.
    }
\label{fig:landscape}
\end{figure}

\subsection{Memory landscape}
\label{sec:results_landscape}

We first examine how the response entropy varies with the location of the causal break and the system--environment coupling.
This sweep uses 20 controlled intervention slots, such that the final readout is at time 2.1.

For each cut, the sampled history contains independent Haar-distributed single-qubit unitaries at all preceding slots.
At the cut, we perform a projective measurement and keep the outcome whose projector is a Haar-distributed qubit state.
The future begins with an independently sampled Haar-distributed pure-state preparation, followed by Haar-distributed unitary controls at the remaining slots.
We use 64 conditioned history settings and 64 future settings, giving a \(256\times64\) response matrix.
Each cut uses one probe draw throughout the coupling sweep, so changes with \(J\) do not arise from a changing probe ensemble.

\Cref{fig:landscape} shows the resulting entropies across cuts and couplings, together with slices through the landscape.

At \(J=0\), panel~\textbf{(a)} recovers the decoupled control where the entropy remains at the rank-one Markovian baseline because the causal break removes history dependence.
Panel~\textbf{(b)} makes the coupling trend explicit, where beyond the earliest cut, \(S_V\) rises throughout the sampled coupling range and begins to level off at stronger coupling.
Panel~\textbf{(c)} shows the temporal trend.
For each displayed nonzero coupling, \(S_V\) rises from the initial boundary into a broad, fluctuating interior plateau before decreasing near the final boundary.
The curves remain ordered throughout, with stronger system--environment coupling producing higher response entropy.
Together, the panels show that the response entropy builds as the cut moves forward in time into the interior and, more distinctly, moves farther from the Markovian baseline as the coupling increases.

The lower values at the first and last cuts require care because moving the cut changes both the evolution time and the lengths of the sampled history and future sequences.
In addition, \(S_V\) characterizes the distribution of response weight rather than memory strength, so these lower values do not by themselves show that memory has weakened or disappeared.
The landscape shows where this weight spreads more broadly, but not whether the spread comes from several comparable modes or from weaker contributions alongside one dominant mode.

\subsection{Response modes}
\label{sec:results_modes}

To distinguish these possibilities, we resolve the individual mode weights and summarize their distribution through the effective mode number \(R\).
This is the number of equally weighted modes that would give the same entropy, rather than an integer count of nonzero singular values.
Values above one can therefore capture contributions beyond the leading mode.

We use 20 intervention slots and the same 64-by-64 Haar-probe construction as the landscape where we draw a fresh ensemble at each cut and coupling.
\Cref{fig:modes} compares \(R\) across cuts and shows selected spectra at an interior cut.
\begin{figure}[]
    \centering
    \includegraphics[width=\columnwidth]{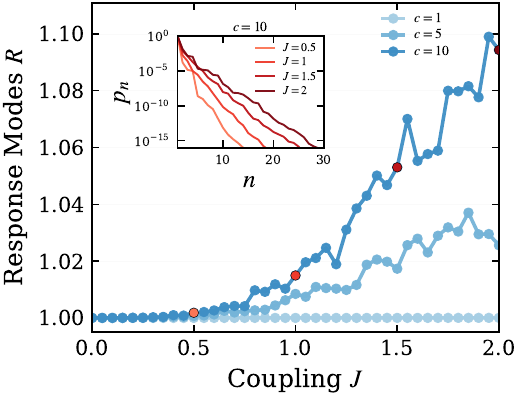}
    \caption{
        \textbf{Response modes and the singular spectrum.} The main panel shows the effective mode number \(R=\exp(S_V)\) as a function of coupling \(J\) for cuts \(c=1,5,10\), with \(L=6,k=20,g=1,\Delta t=0.1\) and \(N_{\mathrm p}=N_{\mathrm f}=64\).
        Each cut and coupling uses an independently drawn Haar-probe ensemble.
        The four highlighted points on the \(c=10\) curve correspond, by coupling and color, to the four lines in the inset at \(J=0.5,1,1.5,2\).
        Each highlighted value of \(R\) uses all retained singular values, while the matching inset line shows the first 30 normalized squared weights \(p_n\) on a logarithmic scale.
        At the interior cuts, increasing coupling raises \(R\) overall and transfers a growing fraction of the response weight into subleading modes.
    }
    \label{fig:modes}
\end{figure}

The main panel highlights the increase in response modes as the coupling between the system and environment increases.
At the earliest cut, \(R\) changes little with coupling, equivalent to what was seen in the memory landscape.
However, at the interior cuts, it instead shows a clear overall increase with coupling, with the larger change occurring at \(c=10\).
Each point uses an independently drawn probe ensemble, so the small point-to-point fluctuations should not be interpreted as a strictly monotonic fixed-probe trend.

Although the plotted values remain relatively close to one, their increase with coupling is clearly resolved.
A value near one means that the leading mode carries most of the squared response weight, but it does not mean that the remaining modes are absent.
This modest absolute scale likely reflects the limited multimode response generated by this finite-time, locally interacting Ising evolution for the chosen probes.
Even in this concentrated regime, the response matrix clearly distinguishes the growing contribution of additional modes.

At the four highlighted points, we plot the response spectrum to show how this increase arises.
At weak coupling, the corresponding inset spectrum falls rapidly beyond its leading weight.
As the coupling increases, the leading mode loses a small fraction of its weight while the subleading tail moves upward and remains visible across more mode indices.
These individually weak contributions thus accumulate in the entropy and produce the increase in \(R\).

Together, the main panel and inset show that stronger coupling redistributes an increasingly visible fraction of the measured response into subleading modes.
One mode still carries most of the weight, so a leading-mode approximation would retain most of the response, but the growing tail resolves richer history-dependent variation.
To determine whether this redistribution reflects the structure of the underlying process, we next compare the response spectrum with the PT-MPO bond spectrum across the same temporal cut.
\begin{figure}[]
    \centering
    \includegraphics[width=\columnwidth]{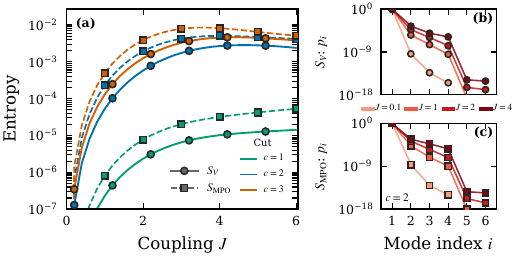}
    \caption{
        \textbf{Response spectrum compared with known PT-MPO structure.} \textbf{(a)} Response entropy \(S_V\) (solid lines, circles) and PT-MPO bond entropy \(S_{\mathrm{MPO}}\) (dashed lines, squares) at cuts \(c=1,2,3\), for \(L=6,k=3,g=1,\Delta t=0.1\).
        Both entropies depart from the rank-one baseline with increasing coupling and show the same cut-dependent growth, including a turnover at the later cuts.
        \textbf{(b,c)} Normalized squared weights obtained from the response singular values and PT-MPO bond values at \(c=2\), for \(J=0.1,1,2,4\), where six modes are shown.
        Both spectra remain dominated by one mode while stronger coupling moves weight into subleading modes.
        Their agreement is structural rather than pointwise, and the response entropy remains lower in this construction.
    }
    \label{fig:mpo_comparison}
\end{figure}

\subsection{Comparison with known process structure}
\label{sec:results_process_spectrum}

The response matrix and a PT-MPO describe the same history--future structure in different representations.
Under the recovery conditions established in \cref{sec:pt-mpo-connection}, their singular values differ only by a common scale, so their normalized spectra and spectral entropies coincide.
We examine this connection using a short Ising evolution whose dense process tensor can be computed directly.
At each temporal cut, we compare the response singular values with the PT-MPO bond values for the same process and history--future division.
Because the exact recovery conditions are not imposed in this calculation, the relevant question is whether the response spectrum follows the PT-MPO structure rather than whether their values agree pointwise.

\Cref{fig:mpo_comparison}(a) shows the same qualitative dependence on coupling and cut in both entropies.
Without coupling, both reduce to the rank-one baseline.
As the coupling increases, both initially rise at every cut.
At the later cuts, both pass through a maximum and then decrease again at stronger coupling, although the PT-MPO entropy reaches its maximum earlier.
Away from the baseline, the later cuts also produce the larger entropies.
The response entropy therefore follows the same coupling- and cut-dependent changes in temporal structure as the PT-MPO bond entropy.

The corresponding spectra in \cref{fig:mpo_comparison}(b,c) show the origin of this agreement.
Both are dominated by one leading mode, while increasing the coupling raises the first few subleading weights above a much weaker tail.
The shared entropy trend therefore reflects the same redistribution of spectral weight rather than an accidental agreement between two scalar summaries.

The agreement is structural rather than pointwise.
The PT-MPO assigns more normalized weight to the subleading modes, so its entropy is larger and reaches its maximum earlier at the interior cuts.
This offset is consistent with the theoretical result because exact equality follows only under its stated recovery conditions.
Together, the entropy and singular-value comparisons show that the response spectrum tracks changes in the known process structure across coupling and temporal cut.
This remains a characterization at a chosen cut and does not determine how long the detected history dependence remains visible once further interventions occur.
The next question is therefore one of persistence of the memory when it is repeatedly reset.
\begin{figure}[tb]
\centering
    \includegraphics[width=\columnwidth]{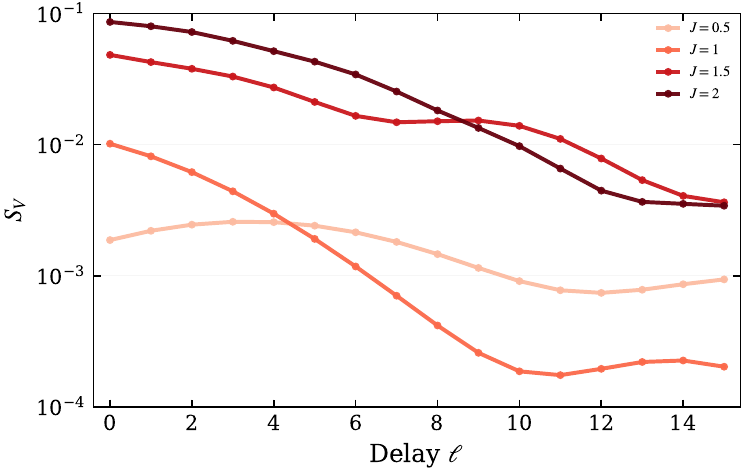}
    \caption{
        \textbf{Response persistence across an outcome-conditioned reset bridge.} Response entropy \(S_V(\ell)\) for \(J=0.5,1,1.5,2\) and bridge lengths \(\ell=0,\ldots,15\).
        A 15-control history ends with an outcome-resolved causal break and preparation of \(\ket{0}\).
        At each of the following \(\ell\) bridge steps, the protocol conditions on the effect \(\ket{0}\bra{0}\) and reprepares \(\ket{0}\).
        A further zero-outcome measurement followed by the sampled preparation begins the five-control future window.
        We use \(L=6,g=1,\Delta t=0.1\) and the same 64 history and 64 future settings throughout.
        Each response is weighted by the probability of the full retained measurement record.
        The bridge generally concentrates the response spectrum, but crossings and short nonmonotonic features show that there is no fixed decay ordering with coupling.
        All curves remain above the rank-one baseline at the reported numerical resolution.
        They describe persistence along the selected measurement record for these fixed probe windows.
    }
    \label{fig:memory_length}
\end{figure}

\subsection{Memory persistence under conditioned resets}
\label{sec:results_memory_length}

To answer this question, we keep the history and future windows fixed and lengthen only the outcome-conditioned reset bridge between them.
The same sampled settings are reused throughout, so the sweep isolates how extending the intervening sequence reshapes a fixed response experiment.

\Cref{fig:memory_length} shows that the response spectrum generally becomes more concentrated as the bridge grows.
At the shortest bridge length, the entropy is ordered by coupling, with stronger system--environment interactions producing the broader response.
For the three larger couplings, the entropy then falls substantially across the sampled window.
The weakest-coupling curve instead broadens over the first few steps before joining the overall decrease.
Repeated resets therefore generally suppress the initial spectral spread across the full bridge, but they do not impose a common decay profile.

The coupling ordering also changes with bridge length.
The \(J=1\) curve falls below the weakly coupled \(J=0.5\) curve, while the initially largest \(J=2\) response eventually falls below the \(J=1.5\) response.
The stronger-coupling curves approach one another at the longest delays, and the \(J=2\) curve is the only one that decreases monotonically across the full window.
The remaining curves contain a short initial rise, a shoulder, or a late upturn.
A broader response at the start of the bridge therefore does not imply that it remains broader for longer.
Instead, the coupling controls both the response structure entering the bridge and how the conditioned intervention sequence reorganizes it.

Because \(S_V\) measures the normalized distribution of response weight rather than its absolute magnitude, its overall decrease means that the leading mode acquires a larger share of the response while subleading modes become relatively less prominent.
Additional reset steps can change different response directions by different amounts, so the nonmonotonic features indicate changing modal visibility rather than the regeneration of memory.
None of the curves reaches the rank-one baseline at the reported numerical resolution, and memory therefore remains visible through every bridge length considered here.

This is a selected-record persistence profile for fixed finite history and future windows.
It does not define an unconditional reset lifetime or establish process-wide Markov order.
Nor does a surviving multimode response identify the carrier as quantum, because a classical feed-forward memory can produce the same qualitative signature.
We therefore close the numerical results with a separate benchmark in which a linear response witness can exclude the stated classical model.

\subsection{Certifying quantum memory from response data}
\label{sec:results_quantum_memory}

The response spectrum establishes whether the tested future remains dependent on the history, but it does not determine whether the carrier of this memory is classical or quantum.
We demonstrate this distinction using the SWAP--reset--SWAP protocol in \cref{fig:swap-memory-schematic}, for which an accessible quantum-memory witness can be obtained analytically and evaluated directly from the response matrix.

\begin{figure}[]
    \centering
    \includegraphics[width=\columnwidth]{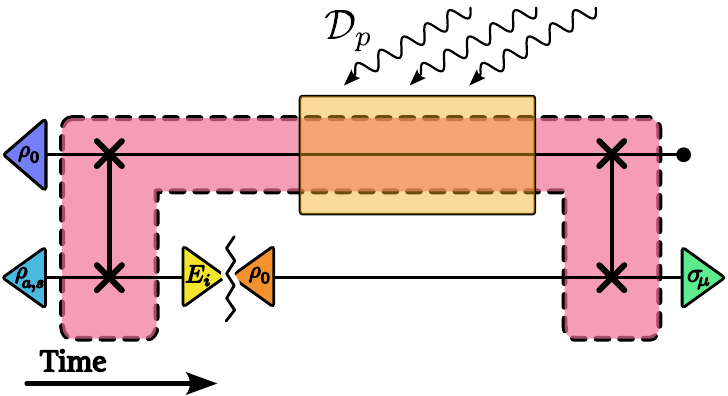}
    \caption{
        \textbf{SWAP--reset--SWAP quantum-memory test.}
        The dashed region is the process tensor.
        The environment is initialized in $\rho_0$ (top wire) and the system is prepared in one of the six Pauli eigenstates $\rho_{a,s}$ (bottom wire).
        In this benchmark $\rho_0=\ket{0}\!\bra{0}$.
        The first SWAP stores $\rho_{a,s}$ in the environment and leaves the system in $\rho_0$.
        At the causal break the system is measured with effect $E_i$ and independently re-prepared in $\rho_0$, where the outcome $i=0$ is deterministic.
        The beige block is the environmental depolarization $\mathcal D_p$, indicated by the dissipative channels.
        The second SWAP returns the surviving state to the system, after which the future probe measures a Pauli observable $\sigma_\mu$ with $\mu\in\{I,X,Y,Z\}$.
        Any dependence of $\sigma_\mu$ on $(a,s)$ must therefore have crossed the break through the environment.
    }
    \label{fig:swap-memory-schematic}
\end{figure}

The accessible system $S$ and environment $E$ are qubits, with the environment initially prepared in $\ket{0}\!\bra{0}_E$.
The system is prepared uniformly in one of the six Pauli eigenstates
\begin{equation}
    \rho_{a,s}=\frac{1}{2}\bigl(I+s\sigma_a\bigr),
    \qquad a\in\{x,y,z\},\qquad s\in\{+1,-1\}.
    \label{eq:swap-input-states}
\end{equation}
The first SWAP stores this state in the environment,
\begin{equation}
    \rho_{a,s}^{S}\otimes\ket{0}\!\bra{0}_E
    \xmapsto{\mathrm{SWAP}}
    \ket{0}\!\bra{0}_S\otimes\rho_{a,s}^{E}.
    \label{eq:first-swap-storage}
\end{equation}
At the causal break, the system is measured in the computational basis, and a fresh system is prepared in $\ket{0}\!\bra{0}$, independently of the outcome.
The corresponding measure--prepare instrument is
\begin{equation}
    \mathcal B_i(\rho)
    =\Tr(E_i\rho)\,\ket{0}\!\bra{0},
    \qquad E_i=\ket{i}\!\bra{i},
    \qquad i\in\{0,1\}.
    \label{eq:swap-causal-break}
\end{equation}
Because the first SWAP leaves $S$ in $\ket{0}$, the outcome $i=0$ occurs with unit probability and $i=1$ has zero probability.
The retained branch is therefore deterministic and requires no postselection.
Meanwhile, the stored state in $E$ crosses the causal cut and undergoes depolarization,
\begin{equation}
    \mathcal D_p(\rho)=(1-p)\rho+p\frac{I}{2},
    \qquad 0\leq p\leq1.
    \label{eq:environment-depolarization}
\end{equation}
The second SWAP returns the surviving state to the system, giving the effective input--output channel
\begin{equation}
    \mathcal E_p=\mathcal D_p,
    \qquad
    \rho^{\mathrm{out}}_{a,s}
    =(1-p)\rho_{a,s}+p\frac{I}{2}.
    \label{eq:swap-effective-channel}
\end{equation}
Since the accessible system was measured, discarded, and re-prepared at the cut, any recovered dependence on $(a,s)$ must have crossed the cut through the environmental process.

For each input $(a,s)$, we measure the identity and three Pauli observables at the output.
The resulting $4\times6$ response matrix has entries
\begin{equation}
    V_{\mu,(a,s)}
    :=\Tr\!\left(
        \sigma_\mu\rho^{\mathrm{out}}_{a,s}
    \right),
    \qquad \mu\in\{I,x,y,z\},
    \label{eq:swap-response-entries}
\end{equation}
where $\sigma_I=I$.
Its six columns are the selected histories, and its four rows are the future observables.
In the notation of \cref{eq:response-probe-operators}, these entries are precisely the accessible contractions $\Tr(B_{xy}\Upsilon_{HF})$, such that no process tensor reconstruction is performed.

The fidelity between an input state and its retrieved output follows directly from the corresponding response column:
\begin{align}
    F_{a,s}
    &:=\Tr\!\left(
        \rho_{a,s}\rho^{\mathrm{out}}_{a,s}
    \right)
      =\frac{1}{2}\left[V_{I,(a,s)}+sV_{a,(a,s)}\right],
      \label{eq:swap-single-fidelity}\\
    F_{\mathrm{av}}
    &:=\frac{1}{6}\sum_{a,s}F_{a,s}
      =\frac{1}{12}\sum_{a,s}
        \left[V_{I,(a,s)}+sV_{a,(a,s)}\right].
      \label{eq:swap-average-fidelity}
\end{align}
We assume trusted preparations and that the classical labels $(a,s)$ are not communicated to the process through an additional side channel \cite{RossetPRX2018}.
Without this assumption, a classical device could store the label and reproduce the requested state perfectly.

\begin{figure}[]
    \centering
    \includegraphics[width=\columnwidth]{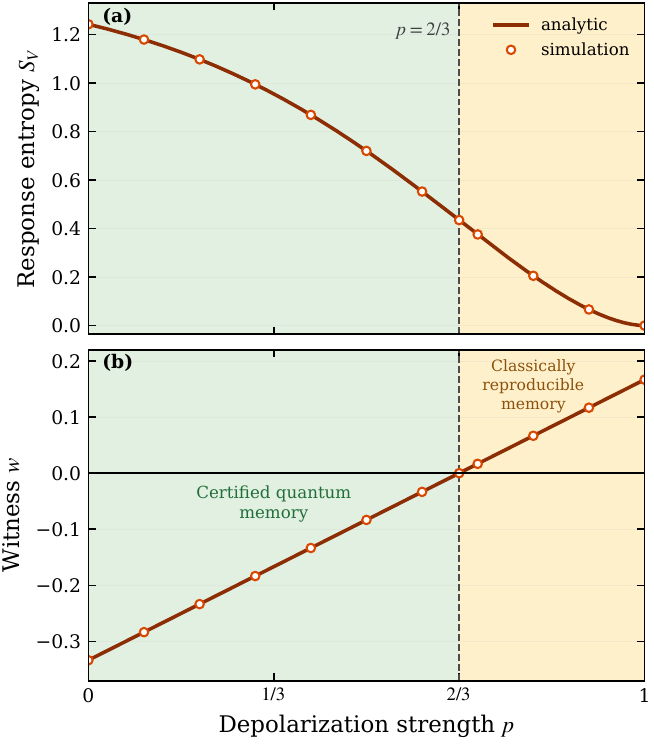}
    \caption{
        \textbf{Certifying quantum memory from response data.}
        In the noisy SWAP--reset--SWAP benchmark, an environmental qubit stores one of six Pauli eigenstate inputs across a system reset and undergoes depolarization of strength \(p\) before retrieval.
        At each \(p\), both panels use the same \(4\times6\) response matrix.
        \textbf{(a)} Response entropy \(S_V\).
        \textbf{(b)} Fixed witness \(w=2/3-F_{\mathrm{av}}\), with mean input--output fidelity \(F_{\mathrm{av}}\).
        Solid lines are analytic predictions, and circles show a subset of the 101 exact density-matrix simulations.
        At \(p=2/3\) (dashed line), the entropy remains positive while the witness crosses zero.
        Negative \(w\) certifies quantum memory under the stated classical feed-forward model, while at the memoryless endpoint \(p=1\), \(S_V=0\).
        The classical classification shown in the panel is specific to this benchmark since a nonnegative witness is generally inconclusive.
    }
    \label{fig:quantum_memory_witness}
\end{figure}

We now establish the classical bound used by the witness.
If only a classical record crosses the causal break, the effective channel has measure--prepare form \cite{Ruskai2003},
\begin{equation}
    \mathcal E_{\mathrm{cl}}(\rho)
    =\sum_\lambda\Tr(M_\lambda\rho)\,\tau_\lambda,
    \qquad M_\lambda\succeq0,
    \qquad \sum_\lambda M_\lambda=I.
    \label{eq:classical-effective-channel}
\end{equation}
Its normalized Choi operator $\omega_{\mathcal E}:=(\mathds{1}\otimes\mathcal E)(\ket{\Phi^+}\!\bra{\Phi^+})$ is separable between the logical input and output.
The six Pauli eigenstates form a qubit state 2-design and satisfy
\begin{equation}
    \frac{1}{6}\sum_{a,s}\rho_{a,s}^{\mathsf T}\otimes\rho_{a,s}
    =\frac{I_4}{6}
     +\frac{1}{3}\ket{\Phi^+}\!\bra{\Phi^+}.
    \label{eq:pauli-six-design-identity}
\end{equation}
Using the normalized-Choi Born rule in \cref{eq:swap-average-fidelity} therefore gives
\begin{equation}
    F_{\mathrm{av}}
    =\frac{1}{3}
     +\frac{2}{3}
      \bra{\Phi^+}\omega_{\mathcal E}\ket{\Phi^+}.
    \label{eq:fidelity-choi-overlap}
\end{equation}
For every normalized separable two-qubit operator, $\bra{\Phi^+}\omega_{\mathrm{sep}}\ket{\Phi^+}\leq1/2$.
The inequality holds for product pure states by the Cauchy--Schwarz inequality and extends to separable mixtures by convexity.
Hence, every classical measure--prepare memory obeys
\begin{equation}
    F_{\mathrm{av}}\leq\frac{2}{3}.
    \label{eq:classical-fidelity-bound}
\end{equation}
This is the single-copy state-estimation bound for the six-state ensemble \cite{MassarPRL1995}.

We consequently define the witness value
\begin{equation}
    w:=\frac{2}{3}-F_{\mathrm{av}}
      =\Tr(W_{\mathrm{fid}}\omega_{\mathcal E}),
    \,
    W_{\mathrm{fid}}
    :=\frac{1}{3}\left(
        I_4-2\ket{\Phi^+}\!\bra{\Phi^+}
    \right).
    \label{eq:swap-fidelity-witness}
\end{equation}
The operator $W_{\mathrm{fid}}$ is a partial-transpose witness of the form introduced in \cref{sec:quantum-memory-certification}:
\begin{equation}
    W_{\mathrm{fid}}
    =\frac{2}{3}
      \left(\ket{\Psi^-}\!\bra{\Psi^-}\right)^{\mathsf T_F},
    \qquad
    \ket{\Psi^-}=\frac{\ket{01}-\ket{10}}{\sqrt{2}}.
    \label{eq:swap-ppt-witness}
\end{equation}
It corresponds to the $P=0$ specialization of \cref{eq:response-witness-sdp}, with $Q=(2/3)\ket{\Psi^-}\!\bra{\Psi^-}$.
The symmetry of this benchmark gives $Q$ analytically, so no SDP search is required.

The same witness is evaluated entirely within the measured response space.
Since $V_{I,(a,s)}=1$ for every deterministic branch, \cref{eq:swap-average-fidelity,eq:swap-fidelity-witness} can be written as the linear response functional
\begin{equation}
    w
    =\frac{1}{36}\sum_{a,s}V_{I,(a,s)}
     -\frac{1}{12}\sum_{a,s}sV_{a,(a,s)}.
    \label{eq:swap-linear-response-witness}
\end{equation}
Thus $W_{\mathrm{fid}}$ lies in the operator span of the selected six history preparations and four future observables.
\Cref{eq:swap-linear-response-witness} is an explicit instance of the response-matrix contraction in \cref{eq:response-witness-contraction}.

Finally, contracting a process separable across the selected $H:F$ cut with the fixed local operations of this protocol can only produce a separable effective Choi operator.
Therefore,
\begin{equation*}
    \Upsilon_{HF}\ \text{separable across }H:F
    \,\Longrightarrow\,
    \omega_{\mathcal E}\ \text{separable}
    \,\Longrightarrow\,
    w\geq0.
\end{equation*}
Its contrapositive proves that $w<0$ certifies entanglement of the full process across the causal cut.
Because the accessible system is discarded and reset at that cut, this entanglement certifies quantum memory carried through the unresolved process rather than direct transmission by the system.

For the depolarizing family, the normalized effective Choi operator is
\begin{equation}
    \omega_p
    =(1-p)\ket{\Phi^+}\!\bra{\Phi^+}+p\frac{I_4}{4}.
    \label{eq:depolarizing-choi-state}
\end{equation}
Its partially transposed minimum eigenvalue and the response witness are
\begin{equation}
    \lambda_{\min}(\omega_p^{\mathsf T_F})=\frac{3p-2}{4},
    \qquad
    F_{\mathrm{av}}=1-\frac{p}{2},
    \qquad
    w=\frac{p}{2}-\frac{1}{3}.
    \label{eq:swap-analytic-results}
\end{equation}
Consequently, within this family,
\begin{equation}
    w<0
    \quad\Longleftrightarrow\quad
    p<\frac{2}{3}
    \quad\Longleftrightarrow\quad
    \omega_p^{\mathsf T_F}\not\succeq0.
    \label{eq:swap-exact-witness-threshold}
\end{equation}
The witness therefore detects the entire entangled regime of this benchmark.

The exact simulations agree with the analytic output states, response spectra, and witness values to numerical precision.
As shown in \cref{fig:quantum_memory_witness}\textbf{(a)}, depolarization progressively concentrates the response spectrum: its three subleading modes lose weight, and $S_V$ decreases continuously.
Nevertheless, $S_V$ remains nonzero at $p=2/3$ and throughout $2/3\leq p<1$, showing that the selected future still depends on the history.

The witness in \cref{fig:quantum_memory_witness}\textbf{(b)} supplies the distinction that the spectrum alone cannot.
For $p<2/3$, its negative value certifies quantum memory.
At $p=2/3$, a measure--prepare strategy reproduces the channel by choosing a Pauli axis uniformly, measuring the input along that axis, storing the axis and outcome as a classical record, and preparing the recorded eigenstate after the break.
This strategy contracts every Bloch vector by $1/3$ and hence implements $\mathcal D_{2/3}$.
Mixing it with complete depolarization realizes $\mathcal D_p$ for $2/3<p<1$, establishing that the remaining memory in this interval is classically reproducible.
At $p=1$, the output is independent of the input, the response becomes rank one, and the process is memoryless.

The classification above the threshold follows from this explicit classical realization, not from the nonnegative witness value: for an arbitrary process, $w\geq0$ remains inconclusive.
The benchmark therefore separates the two roles of the same response data.
Its spectrum characterizes the memory visible to the selected probes, whereas the accessible witness certifies when that memory cannot be carried by the adopted classical feed-forward model.
A complementary approach based on finite-window recovery is developed in Ref.~\cite{An2026}.

\section{Discussion}
\label{sec:discussion}
Here, we introduced a response matrix that organizes how future observations change across controlled histories.
Its spectrum gives a compact description of the structure and persistence of environmental memory, while independently valid linear witnesses can test whether specified classical models of this environment are sufficient.
Under the conditions derived here, the response rank and spectrum match those of the underlying process.
Outside those conditions, they remain an operational summary of the responses that were measured.
The method therefore provides a middle ground between detecting memory and reconstructing every possible multi-time response.

This middle ground changes what can be asked in practice.
Instead of first reconstructing a complete process, one can choose histories and future observations that address a particular model or control question.
Each matrix entry is an independent query, so numerical evaluations can be distributed and extended where more detail is needed.
Applying the same additional queries to a physical process and a learned model would expose temporal discrepancies that one-time observables miss. Such tests could complement scalable learning of Lindblad noise generators from local observable time series \cite{RamosRamos_NoiseLearning_2026}, providing additional checks of whether the learned model captures the observed temporal dependencies.
Agreement would validate the model on those tests, rather than predict unmeasured responses or establish equivalence of the complete processes.

For quantum processors, this creates a direct way to test whether one job has left an environmental trace that changes the next, even after the qubits have been reset.
Identifying which operations leave persistent traces could inform both reset protocols and job scheduling, for example, by separating workloads that would otherwise interfere or by inserting targeted recovery.
Memory characterization can thereby become an input to model validation and control design \cite{Fux2021,Ortega2024,Butler2024} rather than only a diagnosis after failure.

Moving from exact simulations to hardware will require finite-shot uncertainty estimates for spectra and witnesses.
The SDP search remains a theoretical capability, where we leave its numerical implementation and assessment of computational cost and robustness to future work.
Unconditional statements about memory persistence will also require all relevant intervention outcomes rather than the single record studied here.
Developing sampling strategies that recover the required spectral features without measuring every matrix entry remains open, and complete process tensors remain essential when arbitrary multi-time responses are required.
For targeted questions, however, response matrices turn system-only interventions into a scalable probe of hidden dynamics.
They offer a route from detecting environmental memory to learning how to model and control it.

\section*{Acknowledgments}
S. G. thanks Krishna Priya K. M. for the informative discussions on semidefinite programming.
The project leading to this application/publications has received funding from the European Research Council (ERC) under the European Union’s Horizon 2020 research and innovation program (grant agreement No. 101001318).
KM is supported by the Singapore Ministry of Education grant MOE Kick-Starter grant SKI\_2021\_07\_02, and from the Centre for Quantum Technologies grant CQT\_SUTD\_2025\_01.

\section*{Data and code availability}
The data and code used to generate and analyze the numerical results is available from the authors upon reasonable request.

\section*{AI Usage Statement}
During the preparation of this manuscript, the authors used LLMs to assist with language editing and improve the clarity and organization of the text. The authors reviewed and revised the AI-assisted content and take full responsibility for the scientific content and conclusions.

\bibliography{bib.bib}

\clearpage
\onecolumngrid

\appendix
\section{Response-spectrum entropy for a single classical memory bit}
\label{app:classical_memory_entropy}

This example shows how the response spectrum can resolve a classical record
retained by the bath across a causal break. With probes that distinguish its
two values, the response matrix has one mode for each populated branch,
and the response entropy measures the balance of their squared probability weights.
Thus, a nonzero response entropy need not imply quantum memory.

\paragraph{Process and classical record.}
Let $S$ and $E$ be system and bath qubits, initially in
\begin{equation}
    \rho_{SE}(0)=|0\rangle\langle0|_S\otimes
    \sum_{\lambda=0}^{1}p_\lambda|\lambda\rangle\langle\lambda|_E,
    \qquad p_0=\epsilon,\quad p_1=1-\epsilon.
    \label{eq:app_classical_initial_state}
\end{equation}
Here $0\leq\epsilon\leq1$. Two successive interaction intervals are
separated by the causal break. In each interval, the bath applies a CNOT
operation on the system,
\begin{equation}
    U_{SE}=I_S\otimes|0\rangle\langle0|_E
           +X_S\otimes|1\rangle\langle1|_E.
    \label{eq:app_bath_controlled_cnot}
\end{equation}
The bath label $\lambda$ is unchanged, while the conditional system channel
is $\mathcal X_\lambda(A)=X^\lambda A X^\lambda$.
After the first interval,
\begin{equation}
    \rho_{SE}(1)=\sum_{\lambda=0}^{1}p_\lambda
    |\lambda\rangle\langle\lambda|_S\otimes
    |\lambda\rangle\langle\lambda|_E.
    \label{eq:app_state_at_cut}
\end{equation}
This is a classical mixture of product states. The same bath record selects
the system dynamics before and after the break; no coherence between its
two values is required.

\paragraph{Temporal cut and process operator.}
Absorb the fixed initial preparation into the process. Denote the system
output before the break by $H$, the freshly prepared input to the second
interval by $F_{\mathrm{in}}$, and its output by $F_{\mathrm{out}}$.
The future space is $F=F_{\mathrm{in}}\otimes F_{\mathrm{out}}$.
Using input-first Choi ordering, define
\begin{equation}
    J_\lambda=(I\otimes X^\lambda)|\Omega\rangle\langle\Omega|
                        (I\otimes X^\lambda),
    \qquad |\Omega\rangle=|00\rangle+|11\rangle.
\end{equation}
The resulting process Choi operator is
\begin{equation}
    \Upsilon_{HF}=\sum_{\lambda=0}^{1}p_\lambda
    |\lambda\rangle\langle\lambda|_H\otimes J_\lambda^F.
    \label{eq:app_classical_process_choi}
\end{equation}
Each $J_\lambda$ is positive and satisfies
$\operatorname{Tr}_{F_{\mathrm{out}}}J_\lambda=I_{F_{\mathrm{in}}}$.
Equation~\eqref{eq:app_classical_process_choi} therefore gives an explicit
classical feed-forward realization, separable across $H:F$.
Although $J_\lambda$ itself is entangled between the input and output of
the future interval, both legs lie in $F$: that channel entanglement does
not represent quantum memory across the selected cut.

\paragraph{History and future probes.}
At the break, measure $M_i=|i\rangle\langle i|$ and re-prepare
$\sigma=|0\rangle\langle0|$, independently of $i$.
The corresponding operation is
$\mathcal A_i(A)=\operatorname{Tr}(M_iA)\sigma$, with product Choi
operator $M_i^T\otimes\sigma$. Finally, measure $M_j$ after the second
interval. In the response convention of
\cref{eq:response-probe-operators}, the effective probes are
\begin{equation}
    h^{(i)}=M_i,\qquad f^{(j)}=\sigma\otimes M_j,
    \qquad B_{ij}=[h^{(i)}]^T\otimes[f^{(j)}]^T.
\end{equation}
All these matrices are real and diagonal. Their accessible operator space is
\begin{equation}
    \mathcal S_{\mathrm{acc}}
    =\operatorname{span}_{\mathbb R}\{B_{ij}:i,j\in\{0,1\}\}.
\end{equation}
These four probe operators do not span the full process operator space,
but suffice to distinguish the two classical branches.
Contracting them with the process gives
\begin{align}
    (V_c)_{ji}
    &=\operatorname{Tr}(B_{ij}\Upsilon_{HF})\nonumber\\
    &=\sum_{\lambda=0}^{1}p_\lambda
      \operatorname{Tr}(M_iM_\lambda)
      \operatorname{Tr}\!\left[M_jX^\lambda\sigma X^\lambda\right]
      =p_i\delta_{ij}.
    \label{eq:app_classical_joint_probability}
\end{align}
Thus the history outcome identifies the bath bit, and the future outcome
reads the same bit after the system has been reset. These are joint
probabilities, including the probability $p_i$ of the history outcome;
normalizing each history branch separately would remove that information.

\paragraph{Response spectrum and mode weights.}
The response matrix and its singular values are
\begin{equation}
    V_c=\begin{pmatrix}\epsilon&0\\0&1-\epsilon\end{pmatrix},
    \qquad \{\sigma_\ell^{(c)}\}=\{\epsilon,1-\epsilon\}.
    \label{eq:app_classical_response_matrix}
\end{equation}

Consequently, the response entropy of \cref{eq:SV_formal} is
\begin{equation}
    S_V(c) = -u\ln u - (1-u)\ln(1-u),
    \qquad
    u := \frac{\epsilon^2}
              {\epsilon^2 + (1-\epsilon)^2}.
    \label{eq:app_classical_shannon_relation}
\end{equation}
This is the entropy of the normalized squared response weights.
For $0<\epsilon<1$, the response has rank two;
at either deterministic endpoint it has rank one and zero entropy (taking $0\ln0=0$).
The maximum at $\epsilon=1/2$ is $\ln2$, or one bit when logarithms are base two.
Since $\Pr(j\mid i)=\delta_{ij}$ for every outcome of nonzero probability,
$H(j\mid i)=0$ and $I(i:j)=H(i)=-\sum_i p_i\ln p_i$. This mutual information generally differs from the response
entropy $S_V(c)$.
Retrieval is perfect for every $\epsilon$, while $S_V(c)$ describes how evenly
the response weight is distributed between the classical branches.

The resolved modes depend on the chosen probes.
For example, ignoring the history outcome replaces $V_c$ with the
single column $(p_0,p_1)^T$, whose only normalized singular value is one
and whose response entropy is zero, although the bath process is unchanged.
The response spectrum therefore characterizes memory accessible to the
specified probes. In this example, that memory is entirely classical, as
established independently by \cref{eq:app_classical_process_choi}; its
nonzero entropy is not a quantum-memory certificate.

\section{Partial-transpose witnesses and the choice of probes}
\label{app:quantum-memory-witness}

The transpose provides a concrete positive map for the witness search
in \cref{eq:response-witness-sdp}. Choosing
\(\Lambda_F=T_F\), which is self-adjoint with respect to the
Hilbert--Schmidt inner product, gives
\begin{equation}
    W_C=P+Q^{T_F},
    \qquad P,Q\succeq0.
    \label{eq:appendix-decomposable-witness}
\end{equation}
For every separable positive operator \(\Sigma_{HF}\), its partial
transpose remains positive, so
\(\operatorname{Tr}(W_C\Sigma_{HF})\geq0\).
A negative expectation therefore certifies entanglement across the
selected temporal cut. This family cannot detect entangled operators
with positive partial transpose: for such operators, both contributions
to the witness expectation remain nonnegative.

\paragraph{Connection to the eigenvector construction.}
Setting \(P=0\) gives \(W_C=Q^{T_F}\).
Without the accessible-space restriction, the unit-trace search becomes
\begin{equation}
    \min_{Q\succeq0,\ \operatorname{Tr}Q=1}
    \operatorname{Tr}(Q\UpsilonPT_{HF}^{T_F})
    =
    \lambda_{\min}(\UpsilonPT_{HF}^{T_F}).
    \label{eq:appendix-ppt-eigenvalue}
\end{equation}
Indeed, the objective is an average of expectation values of
\(\UpsilonPT_{HF}^{T_F}\), and its minimum is attained by placing
all weight on a lowest-eigenvalue eigenvector.
If this eigenvector is \(\lvert\nu\rangle\), then
\(Q=\lvert\nu\rangle\langle\nu\rvert\) yields the witness
\(W_C=(\lvert\nu\rangle\langle\nu\rvert)^{T_F}\).
This recovers the partial-transpose SDP used by
Giarmatzi and Costa~\cite{Giarmatzi2021}.
In our setting, one must additionally require
\(Q^{T_F}\in\mathcal S_{\mathrm{acc}}\), since only witnesses in this
space can be evaluated from the response matrix.

Although \(P\succeq0\) contributes a nonnegative expectation on every
physical process, retaining it can nevertheless improve an accessible
witness search. The accessibility constraint applies to the sum
\(P+Q^{T_F}\), not to either term separately. Their inaccessible
components can cancel, producing a measurable witness that would be
excluded by the restriction \(P=0\). Thus, without probe restrictions,
\(P=0\) already detects every negative partial transpose; with restricted
probes, allowing \(P\) can be essential.

\paragraph{Choice of normalization.}
The constraint
\(\operatorname{Tr}P+\operatorname{Tr}Q\leq1\) both bounds the witness
scale and includes the zero witness. Consequently, the optimum is
never positive and is exactly zero for every cut-separable process.
If a negative optimum exists, its trace constraint is saturated:
otherwise, rescaling all variables by a common positive factor would
produce a more negative feasible value.

One may instead impose
\(\operatorname{Tr}P+\operatorname{Tr}Q=1\) when a fixed normalization
is desired, for example to recover
\cref{eq:appendix-ppt-eigenvalue} or compare witnesses within the same
normalized family. Provided the normalized feasible set is nonempty,
this does not change whether a negative witness exists.
It does change the interpretation of a failed search: the optimum
may now be positive, and both zero and positive values are inconclusive.
A restricted probe span can also make this equality-constrained problem
infeasible. We therefore retain the inequality in the main text.

\paragraph{An example of probe dependence.}
Consider two qubit factors, one on each side of the temporal cut, and
the normalized operators
\begin{equation}
    \rho_{\mathrm{ent}}
    =
    \lvert\Phi^+\rangle\langle\Phi^+\rvert,
    \qquad
    \rho_{\mathrm{sep}}
    =
    \frac12
    \left(
      \lvert00\rangle\langle00\rvert+
      \lvert11\rangle\langle11\rvert
    \right).
    \label{eq:appendix-probe-example-states}
\end{equation}
These operators illustrate the entanglement test on the physical
\(H:F\) factors; they are not assumed to represent an arbitrary full
process tensor. A temporal realization is a process that emits the two
members of an initially stored pair before and after the break,
respectively, while ignoring the freshly prepared system input.
The corresponding input identity factors can then be suppressed in
this example.

Suppose first that the available local probes are the computational-basis
projectors. Their real span is
\(\operatorname{span}_{\mathbb R}\{I,Z\}\), and hence
\begin{equation}
    \mathcal S_Z
    =
    \operatorname{span}_{\mathbb R}
    \{I\otimes I,\ I\otimes Z,\ Z\otimes I,\ Z\otimes Z\}.
\end{equation}
The two operators in
\cref{eq:appendix-probe-example-states} give identical expectations
for all four basis elements: respectively \(1,0,0,1\).
Therefore, they produce identical response data for these probes.
Any accessible witness that is nonnegative on \(\rho_{\mathrm{sep}}\)
must also be nonnegative on \(\rho_{\mathrm{ent}}\).
Here the absence of a certificate follows from a precise information
limitation: the probes do not distinguish the entangled operator from
a separable one.

Now include projectors onto the \(X\) eigenstates on each side.
The accessible space expands to
\begin{equation}
    \mathcal S_{XZ}
    =
    \operatorname{span}_{\mathbb R}
    \{A\otimes B:\ A,B\in\{I,X,Z\}\}.
\end{equation}
Within this space, consider
\begin{equation}
    W_{XZ}
    =
    \frac14
    \left(I\otimes I-X\otimes X-Z\otimes Z\right).
    \label{eq:appendix-xz-witness}
\end{equation}
For a product of normalized qubit states with Bloch vectors \(a,b\),
its expectation is
\((1-a_xb_x-a_zb_z)/4\geq0\), by the
Cauchy--Schwarz inequality and the unit bound on Bloch-vector lengths.
Linearity extends this bound to every separable positive operator.
However,
\begin{equation}
    \operatorname{Tr}(W_{XZ}\rho_{\mathrm{ent}})
    =-\frac14,
\end{equation}
so the additional probes reveal entanglement without requiring full
tomography.

This same example shows why retaining \(P\) matters.
Every operator in \(\mathcal S_{XZ}\) is unchanged by partial
transposition on \(F\), since \(I,X,Z\) are real symmetric matrices.
Consequently, an accessible witness of the restricted form
\(W_C=Q^{T_F}\) must satisfy
\(Q=W_C^{T_F}=W_C\succeq0\), and cannot have a negative expectation on
any positive process operator.

Nevertheless, \(W_{XZ}\) is feasible in the larger search:
\begin{equation}
    P=Q=\frac12\lvert\Psi^-\rangle\langle\Psi^-\rvert,
    \qquad
    \lvert\Psi^-\rangle
    =
    \frac{\lvert01\rangle-\lvert10\rangle}{\sqrt2},
    \qquad
    P+Q^{T_F}=W_{XZ}.
\end{equation}
Here \(\operatorname{Tr}P+\operatorname{Tr}Q=1\).
The inaccessible \(Y\otimes Y\) contributions in the two terms cancel,
leaving the accessible witness in
\cref{eq:appendix-xz-witness}.
Thus the choice of probes controls which witnesses can be evaluated,
while retaining \(P\) can enlarge the detectable family even within
a fixed probe space.

\section{Numerical consistency checks}
\label{app:numerical-consistency}

The main benchmarks use a fixed probe budget and a finite spin environment. Here, we report the
supporting checks that determine how the response entropy changes with these numerical choices.

\begin{figure}[]
    \centering
    \includegraphics[width=0.6\textwidth]{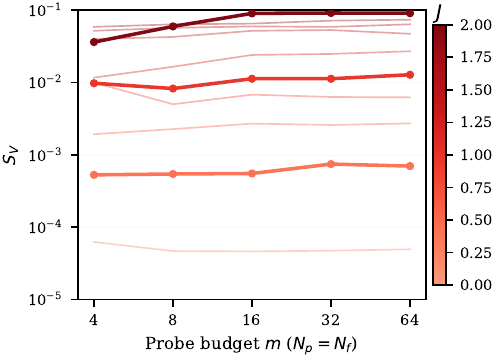}
    \caption{
    \textbf{Probe-budget dependence of the cross-cut memory entropy.}
    Cross-cut memory entropy \(S_V\) at fixed cut \(c=10\) as a function of the symmetric probe
    budget \(m\), with \(N_{\mathrm p}=N_{\mathrm f}=m\). Faint curves show the full coupling sweep,
    while highlighted curves correspond to representative couplings \(J=0.4,1,2\). The
    coupling-dependent ordering is visible at modest budgets, and the displayed changes generally
    become smaller as \(m\) increases. The main benchmarks use the largest tested value,
    \(N_{\mathrm p}=N_{\mathrm f}=64\).
    }
    \label{fig:convergence}
\end{figure}

\subsection{Probe-budget dependence}
\label{app:probe-budget-dependence}

A central practical question is how many sampled past and future probes are needed to resolve the
cross-cut memory structure. To address this, we vary the probe budget symmetrically as
\(N_{\mathrm p}=N_{\mathrm f}=m\) and monitor the resulting entropy \(S_V\) at fixed cut \(c=10\)
across the coupling sweep. \Cref{fig:convergence} shows the resulting dependence on \(m\), with faint
curves indicating the full set of sampled couplings and highlighted curves showing representative
values \(J=0.4,1,2\).

The coupling-dependent ordering is already visible at modest probe budgets. Weakly interacting cases
remain in the low-entropy regime, while stronger interactions produce larger values of \(S_V\).
Increasing \(m\) mainly refines the quantitative estimate rather than changing this ordering. The
largest changes occur at small probe budgets; beyond roughly \(m\sim 8-16\), the representative
curves vary more slowly and approach a broad plateau on the logarithmic scale shown in
\cref{fig:convergence}.

Within this fixed-cut sweep, a finite sampled ensemble resolves the dominant singular-spectrum
broadening, while larger probe budgets change the estimate less. This does not establish convergence
to an informationally complete tester or a protocol-independent value. It shows only that the
qualitative ordering is stable over the larger budgets tested here.

The default choice \(N_{\mathrm p}=N_{\mathrm f}=64\) is the largest budget in this fixed-cut test.
We therefore use it for the memory landscape and mode analysis, without claiming that the same
budget is sufficient for other probe distributions, cuts, or processes.

\begin{figure}[]
    \centering
    \includegraphics[width=\textwidth]{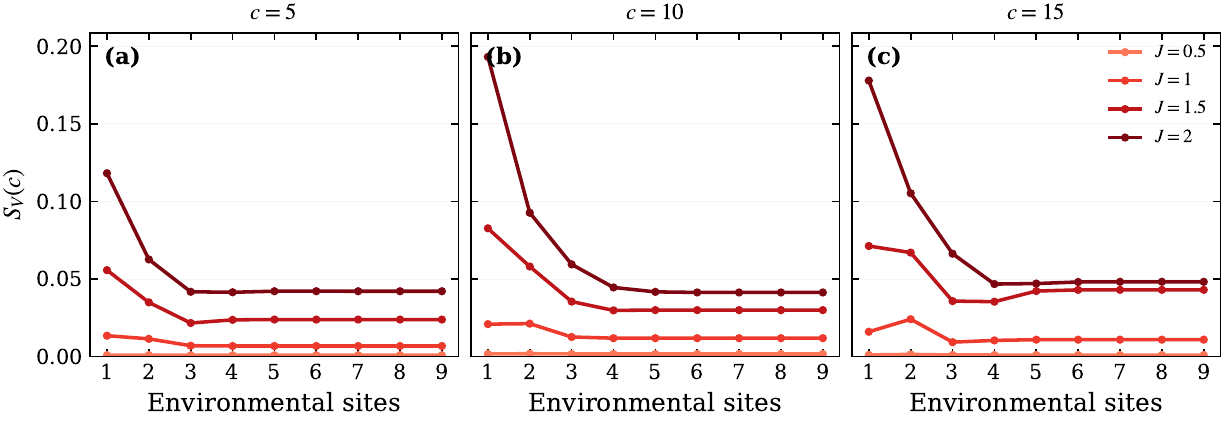}
    \caption{
    \textbf{Finite-environment consistency of fixed-cut memory.}
    Cross-cut memory \(S_V(c)\) at representative cuts as a function of the number of environmental
    sites for couplings \(J=0.5,1,1.5,2\).
    \textbf{(a)} \(c=5\).
    \textbf{(b)} \(c=10\).
    \textbf{(c)} \(c=15\).
    Weak coupling produces little accessible memory for all environment sizes studied. At stronger
    coupling, the extracted memory is largest for the smallest environments and changes less once
    several environmental sites are included. The strongest size dependence in this one-system-site
    benchmark occurs among the smallest environments sampled.
    }
    \label{fig:finite_size}
\end{figure}

\subsection{Finite-environment consistency}
\label{app:finite-environment-consistency}

We next assess how the fixed-cut entropy depends on environment size. For the one-system-site
benchmark considered here, we vary the number of environmental sites and evaluate \(S_V(c)\) at
representative cuts and couplings. \Cref{fig:finite_size} shows the resulting dependence.

The strongest size dependence occurs for very small environments. At weak coupling, the entropy
remains close to the low-memory regime for all environment sizes studied. At stronger coupling,
\(S_V(c)\) is largest for one or two environmental sites, then decreases and approaches a broad
plateau once several environmental sites are included. Thus increasing the environment does not lead
to continued growth of the extracted memory over the simulated time window. Instead, the data show a
crossover from a strongly finite-environment regime to a more stable extended-environment regime.

One possible interpretation follows from the finite observation time and the single-system-site
geometry. In very small environments, information transferred out of the system can return
efficiently, producing a comparatively large cross-cut response. With more environmental sites, it
can spread over more degrees of freedom and the displayed fixed-cut entropy changes less. The
nonzero values at the larger sampled sizes show that the detected response is not restricted to a
one-site environment within this benchmark.

This check should be interpreted within the benchmark geometry. Because the present simulations use
one system site coupled to an environment, only a modest number of environmental sites is needed
before the selected cuts become comparatively stable. For larger accessible systems, or for
geometries with more system--environment contact points, one should generally expect larger
environments to be required before analogous finite-environment plateaus are reached.

\end{document}